\documentclass[11pt,a4paper,notitlepage]{article}

\usepackage{amsmath}        % extensions for typesetting of math
\usepackage{amsfonts}       % math fonts
\usepackage{amsthm}         % theorems, definitions, etc.
\usepackage{bbding}         % various symbols (squares, asterisks, scissors, ...)
\usepackage{bm}             % boldface symbols (\bm)
\usepackage{graphicx}       % embedding of pictures
\usepackage[square,sort,comma,numbers]{natbib}         % citation style AUTHOR (YEAR), or AUTHOR [NUMBER]
\usepackage[nottoc]{tocbibind} % makes sure that bibliography and the lists of figures/tables are included in the table of contents
\usepackage{paralist}       % improved enumerate and itemize
\usepackage[usenames]{xcolor}  % typesetting in color

\usepackage{xfrac}
\usepackage{tikz}
\usepackage{tikz-cd}
\usetikzlibrary{arrows,shapes,positioning}
\usetikzlibrary{decorations.markings}
\usetikzlibrary{arrows, matrix}

\usepackage{wrapfig}
\usepackage{caption}
\usepackage{titlesec}
\usepackage{filecontents}
\usepackage{mathtools}
\usepackage{mdwlist}
\usepackage{mathrsfs}
\usepackage[shortlabels]{enumitem}
\usepackage{amssymb}
\usepackage{textcomp}
\usepackage[toc,page]{appendix}
\usepackage{pgfplots}
\pgfplotsset{compat=1.7}
\usepackage[utf8]{inputenc} % To interpret Czech accents and other characters
\usepackage{authblk} % To write the affiliations
\usepackage[cal=boondox,calscaled=1.0]{mathalfa}
\usepackage{physics}
\usepackage{float}

\numberwithin{equation}{section}
\numberwithin{figure}{section}
\theoremstyle{definition}
\newtheorem{definition}{Definition}[section]
\newtheorem{observation}{Observation}[section]

\newtheorem{theorem}{Theorem}[section]

\newtheorem{example}{Example}[section]

\usepackage{natbib}
\usepackage[nottoc]{tocbibind}

\title{Quantum Mechanical Observables under a Symplectic Transformation of Coordinates}
\author[1]{Jakub Káninský}

\affil[1]{Charles University, Faculty of Mathematics and Physics, Institute of Theoretical Physics. E-mail address: jakubkaninsky@seznam.cz}

\date{\today}
\titleformat{\section}{\normalfont\scshape\large}{\thesection}{2em}{}
\titleformat{\subsection}{\normalfont\scshape\large}{\thesubsection}{1em}{}
\titleformat{\subsubsection}{\itshape\large}{\thesubsubsection}{1em}{}

\renewenvironment{abstract}
 {\small
  \begin{center}
  \end{center}
  \list{}{
    \setlength{\leftmargin}{15mm}%
    \setlength{\rightmargin}{\leftmargin}%
  }
  \item\relax}
 {\endlist}

\begin{document}

\maketitle

\begin{abstract}
We consider a general symplectic transformation (also known as linear canonical transformation) of quantum-mechanical observables in a quantized version of a finite-dimensional system with configuration space isomorphic to $ \mathbb{R}^{q} $. Using the formalism of rigged Hilbert spaces, we define eigenstates for all the observables. Then we work out the explicit form of the corresponding transformation of these eigenstates. A few examples are included at the end of the paper.
\end{abstract}

\vspace{\baselineskip}
\vspace{\baselineskip}

\section{Introduction}
From a mathematical perspective we can view quantum mechanics as a science of finite-dimensional quantum systems, i.e. systems whose classical pre-image has configuration space of finite dimension. The single most important representative from this family is a quantized version of the classical system whose configuration space is isomorphic to $ \mathbb{R}^{q} $ with some $ q \in \mathbb{N} $, which corresponds e.g. to the system of finitely many coupled harmonic oscillators. Classically, the state of such system is given by a point in the phase space, which is a vector space of dimension $ 2q $ equipped with the symplectic form. One would typically prefer to work in the abstract setting, without the need to choose any particular set of coordinates: in principle this is possible. In practice though, one will often end up choosing a symplectic basis in the phase space aligned with the given symplectic form and resort to the coordinate description. Needless to say, there are many equivalent choices of such basis, all mutually related through symplectic transformations.

There is a number of ways to introduce the quantum counterpart of the above system; let us recall two. The first one originates in many body quantum mechanics and comes about by complexifying the classical phase space and performing polarization, which results in the so called one-particle Hilbert space. The polarization is not unique, but its different choices lead to unitarily equivalent theories. The one-particle Hilbert space is then used to build a symmetric Fock space which accommodates the states of our quantum system. The observables are defined with the help of the symplectic form by virtue of creation and annihilation operators. The construction is performed algebraically without any reference to coordinates, which makes it relatively elegant. For further details we refer to Chapter 2 of \cite{Wald1994}. Another option---and a more common one---is the canonical construction in Schr\"{o}dinger representation. It consists in introducing the Hilbert space of square-integrable complex functions on the configuration space and defining the coordinate and momentum observables straightforwardly as multiplication and differentiation operators. Note that the word \textit{configuration space} refers to an implicit choice of symplectic basis needed to fix the representation. The two briefly described constructions can be related by the so called Bargmann transform, see the classical paper \cite{Bargman1961} and Chapter 4 of \cite{Neretin2011}. The isomorphism between the Fock and Bargmann space is proven in \cite{Stochel1997}. Without going to further details, we would like to point out the following fact. In the many-body Fock space construction, one makes a choice of polarization which splits the complexified phase space in half, while in the canonical construction one splits in half the original phase space by a choice of symplectic basis. In any case, the resulting quantum theory is independent of these choices, as we know thanks to Stone-von Neumann Theorem. More elegant approaches are available e.g. within Geometrical Quantization \cite{Carosso2017} or other abstract framework.

Let us quantize our finite-dimensional classical system by means of the canonical procedure. As mentioned above, it involves a choice of symplectic basis in the phase space allowing to establish a well defined complete set of commuting observables. Computations involving quantum states can then be performed in the eigenstate basis associated to the chosen set of commuting observables. A small technical complication is given by the fact that the canonical observables do not have eigenvectors in the original Hilbert space: they are not square integrable functions. One can work around this by introducing a generalized version of the eigenvectors from the realm of distributions. In the present paper, they will be called simply \textit{eigenstates}. One can then directly define eigenstates which correspond to states of the system with a sharp value of the associated observable. There are vectors in the Hilbert space that are arbitrarily close to that.

We have pointed out that quantum systems constructed upon different choices of the symplectic basis are equivalent, mutually related by a symplectic transformation of observables. It induces a corresponding transformation of eigenstates, which is in general nontrivial. Meanwhile, having access to eigenstates of symplectically transformed observables can be very useful for practical purposes. For instance, it may happen that one starts with quantum observables defined in one symplectic basis and later decides to change to another one, e.g. in order to simplify the evolution equations or for any other reason. A typical example of this is switching between the coordinate and momentum representations, which is arguably the most prominent special case of symplectic transformation. As common knowledge tells us, the eigenstates of coordinates and momenta are related by the Fourier transform. However, more general cases are not always easily accessible in the literature. The aim of the present paper is to provide a fully general prescription for the transformation of eigenstates corresponding to an arbitrary symplectic transformation of observables.

We need to remark that symplectic transformations of quantum mechanical observables are well known and have been studied before in various contexts. They belong to the broader family of \textit{canonical transformations} which are defined by their property of leaving the canonical commutation relations invariant \cite{Anderson1994}. In this work, we are only interested in their subset referred to as \textit{linear canonical transformations} which have been historically the subject of extensive research not only in quantum physics but also in optics, because they can be used to describe propagation of light rays. For this historical account, see Chapter 1 of \cite{Healy2016}. The basic definition and properties of linear canonical transformations are given in Chapter 2 of the same reference. A more quantum-mechanical point of view is taken in Chapter 9 of \cite{Wolf1979}. In these, one can find prescriptions for integral kernels of the transformation for the case of two dimensions or various special kinds of the transforms in $ 2q $ dimensions (often those relevant in optics). In principle, these can be composed to obtain more general transforms. However, we prefer to present the problem in quantum-mechanical terms and work out the eigenstates of transformed observables from first principles, without any unnecessary outer input. In doing so, we pay attention not only to the transform itself but also to the quantum mechanical theory. The term \textit{symplectic transformation} is used because we feel it better describes the nature of the transform and emphasizes the important relation of the quantum and the classical system, as already discussed above.

The paper is organized as follows. In Sec. 2 we describe in the necessary detail both the classical and quantum versions of our system and introduce rigged Hilbert spaces which will allow us to take advantage of the Dirac formalism. In Sec. 3 we briefly review the needed algebraic tools. Sec. 4 is dedicated to the computation itself, which will provide general formulas for the eigenstates of symplectically transformed observables. We also examine the resulting wavefunctions. Finally, in the last Sec. 5 we give a couple of examples to illustrate the application of the results.

\section{The Classical System and the Quantum System}
In this section we give a formal description of our quantum system and establish the formalism for the rest of the paper. Let us start with a finite-dimensional classical system whose configuration space $ \mathcal{Q} $ has the natural structure of a vector space, i.e., one may identify $ \mathcal{Q} = \mathbb{R}^{q} $. Then the phase space $ \mathcal{P} = \mathbb{R}^{2q} $ of the system is also a vector space equipped with a symplectic form $ \omega : \mathcal{P} \times \mathcal{P} \rightarrow \mathbb{R} $. Since we will not be interested in the evolution, these are really all the classical structures we need. We take advantage of the basic construction with Hilbert space $ \mathcal{F} = L^{2}(\mathcal{Q}) $ of square-integrable complex functions on $ \mathcal{Q} $ with the usual inner product $ ( ~ , ~ ) : \mathcal{F} \times \mathcal{F} \rightarrow \mathbb{C} $ given by
\begin{equation}\label{key}
(\psi,\varphi) = \int_{\mathbb{R}^{q}} \overline{\psi(x)} \varphi(x) ~ d^{q}x
\end{equation}

We shall choose a symplectic basis $ \{ e_{I} \}_{I = 1}^{2q} $ in $ \mathcal{P} $, so that we can write a vector $ y \in \mathcal{P} $ in coordinates as $ y = y_{I} e_{I} $ with implicit summation over $ I = 1, ..., 2q $. We may then identify the configuration space $ \mathcal{Q} $ with the space spanned by $ \{ e_{A} \}_{A = 1}^{q} $, so that $ x \in \mathcal{Q} $ is written as $ x = x_{A} e_{A} $. This identification of the configuration space is natural, but we must not forget that it is basis-dependent. Next, recall that the assumption of $ \{ e_{I} \}_{I = 1}^{2q} $ being symplectic means
\begin{equation}\label{key}
\omega(e_{I}, e_{J}) = \sigma_{IJ}
\end{equation}
with a $ 2q \times 2q $ matrix $ \sigma $ of the block structure
\begin{equation}\label{key}
\sigma \equiv \begin{pmatrix}
\phantom{-} 0 & \mathbf{1} \\
- \mathbf{1} & 0
\end{pmatrix}
\end{equation}
The symplectic product of two vectors $ y, u \in \mathcal{P} $ then has a simple coordinate form
\begin{equation}\label{symplprod}
\omega(y,u) = y_{A} u_{A+q} - y_{A+q} u_{A}
\end{equation}
\iffalse
Note that we use an opposite sign convention for $ \omega $ compared to that in \cite{Wald1994}.
\fi

Now we introduce the coordinate and momentum operators used to fix the representation. We draw from the discussion in \cite{Madrid2005}. For further details, as well as a general treatment of operators on Hilbert spaces, we refer to the canonical book \cite{Reed1981} and the lecture notes \cite{Landsman2006}. We start with the \textit{coordinate operator} $ \hat{y}_{A} : \mathcal{D}(\hat{y}_{A}) \rightarrow \mathcal{F} $ which shall be given for all $ \varphi \in \mathcal{D}(\hat{y}_{A}) \subset \mathcal{F} $ by
\begin{equation}\label{ymult}
\hat{y}_{A} \varphi(x) = x_{A} \varphi(x)
\end{equation}
where $ x \in \mathcal{Q} $. Note that this definition asserts that the domain $ \mathcal{D}(\hat{y}_{A}) $ of $ \hat{y}_{A} $ must be such that $ \hat{y}_{A} \varphi(x) \in \mathcal{F} $, i.e., we have $ \mathcal{D}(\hat{y}_{A}) = \{ \varphi \in L^{2}(\mathcal{Q}) ~ \vert ~ \int_{\mathbb{R}^{q}} \vert x_{A} \varphi(x) \vert^{2} ~ d^{q}x < \infty \} $. One can find that $ \mathcal{D}(\hat{y}_{A}) $ is not the whole $ \mathcal{F} $ (though it is dense in $ \mathcal{F} $) and $ \hat{y}_{A} \mathcal{D}(\hat{y}_{A}) $ is not included in $ \mathcal{D}(\hat{y}_{A}) $. Nevertheless, we continue by defining the \textit{momentum operator} $ \hat{y}_{A+q} : \mathcal{D}(\hat{y}_{A+q}) \rightarrow \mathcal{F} $ via
\begin{equation}\label{yder}
\hat{y}_{A+q} \varphi(x) = - i \frac{\partial}{\partial x_{A}} \varphi(x)
\end{equation}
As in the case of $ \hat{y}_{A} $, the domain $ \mathcal{D}(\hat{y}_{A+q}) $ of $ \hat{y}_{A+q} $ is only dense in $ \mathcal{F} $ and is not invariant under the action of $ \hat{y}_{A+q} $. In general terms of analysis on Hilbert spaces, one can say that both $ \hat{y}_{A}, \hat{y}_{A+q} $ are unbounded, their spectrum is the whole real line, and they do not have any eigenvectors in $ \mathcal{F} $. Along with that, expectation values of these operators are not finite and algebraic operations such as commutation relations involving these operators are not well defined on the whole $ \mathcal{F} $.

Before we go on to resolve the domain problems of the coordinate and momentum operators, let us choose a representation $ \pi : \mathcal{P} \rightarrow \mathcal{U}(\mathcal{F}) $ of the canonical commutation relations, i.e., a map from the phase space $ \mathcal{P} $ to the space $ \mathcal{U}(\mathcal{F}) $ of unitary operators on $ \mathcal{F} $ satisfying the Weyl relations
\begin{equation}\label{weyl}
	W^{\pi}(y) W^{\pi}(u) = e^{ - i \frac{1}{2} \omega(y,u) } ~ W^{\pi}(y+u)
\end{equation}
\begin{equation}\label{key}
	W^{\pi}(y)^{\dagger} = W^{\pi}(-y)
\end{equation}
For details, see \cite{Derezinski2005}. We opt for the \textit{Schr\"{o}dinger representation} defined by
\begin{equation}\label{Wy}
W(y) = e^{ i \left( y_{A} \hat{y}_{A+q} - y_{A+q} \hat{y}_{A} \right) }
\end{equation}
The resemblance of the exponent in \eqref{Wy} to \eqref{symplprod} is no coincidence. Our choice results in the self-adjoint field operators
\begin{equation}\label{phi}
\phi(y) \equiv \hat{\omega}(y, \cdot ) = y_{A} \hat{y}_{A+q} - y_{A+q} \hat{y}_{A}
\end{equation}
which can be interpreted as the observables associated to the classical linear functions $ \omega(y, \cdot) : \mathcal{P} \rightarrow \mathbb{R} $ that take $ y \in \mathcal{P} $ as a parameter and map $ u \mapsto \omega(y, u) = y_{A} u_{A+q} - y_{A+q} u_{A} $. Note that, in particular, \eqref{ymult} and \eqref{yder} themselves are field operators with $ \hat{y}_{A} = \phi(-e_{A+q}) $ and $ \hat{y}_{A+q} = \phi(e_{A}) $. These are especially useful: according to our physical interpretation, they correspond to coordinates $ y_{A} $ and momenta $ y_{A+q} $, respectively, of a vector $ y \in \mathcal{P} $ in the phase space of the classical system.

\vspace{\baselineskip}

To be able to work safely with the coordinate and momentum operators, we shall use the formalism of rigged Hilbert spaces described concisely in \cite{Madrid2005}. For background on the topic we refer to \cite{Gelfand1964} and \cite{Gadella2002}. A rigged Hilbert space is a triad of spaces $ \Phi \subset \mathcal{F} \subset \Phi^{\times} $ where $ \mathcal{F} $ is a Hilbert space (we plug in our choice straight away), $ \Phi $ is a dense subset of $ \mathcal{F} $ and $ \Phi^{\times} $ is the space of antilinear functionals over $ \Phi $. We define
\begin{equation}\label{Phidef}
\Phi = \bigcap_{\mathfrak{m}} ~ \mathcal{D} \left( \hat{y}_{1}^{m_{1}} ... \hat{y}_{2q}^{m_{2q}} \right)
\end{equation}
where the intersection is taken over all $ m_{I} = 0, ..., \infty $ in the multiindex $ \mathfrak{m} = ( m_{1}, ..., m_{2q} ) $. That is, $ \Phi $ accommodates test functions which are in the domain of any product of the position and momentum operators, so that expectation values and commutation relations of these operators are well defined on $ \Phi $. Moreover, $ \Phi $ turns out to be the largest subdomain of $ \mathcal{F} $ that \textit{remains invariant} under the action of any power of $ \hat{y}_{I} $. Thanks to this invariance, the expectation values $ (\varphi, \hat{y}_{I}^{m_{I}} \varphi) $ for $ \varphi \in \Phi $ are finite. Likewise, for coordinate and momentum operators narrowed to $ \Phi $, the relation \eqref{weyl} yields
\begin{equation}\label{ccr}
[ \hat{y}_{I} , \hat{y}_{J} ] = i \sigma_{IJ} \hat{\mathbf{1}}
\end{equation}
as an identity on $ \Phi $. In particular, it is $ [ \hat{y}_{A} , \hat{y}_{A+q} ] = i \hat{\mathbf{1}} $, which is the familiar commutation relation characteristic of coordinate and momentum operators in quantum mechanics.

The definition \eqref{Phidef} amounts to saying that the norm $ \| ~ \|_{\mathfrak{m}} $ defined by
\begin{equation}\label{mnorm}
\| \varphi \|_{\mathfrak{m}} = \left( \int_{\mathbb{R}^{q}} \vert \hat{y}_{1}^{m_{1}} ... \hat{y}_{2q}^{m_{2q}} \varphi(x) \vert^{2} ~ d^{q}x \right)^{1/2} 
\end{equation}
is finite for all $ \varphi \in \Phi $. An investigation of \eqref{mnorm} identifies $ \Phi $ to be the Schwartz space $ \Phi = \mathcal{S}(\mathbb{R}^{q}) $ of smooth rapidly decreasing functions on $ \mathbb{R}^{q} $, see e.g. the reference \cite{Reed1981}. It can be checked that the operators $ \hat{y}_{I} $, although not bounded---or to say, continuous---with respect to the $ L^{2} $-norm, are bounded with respect to the norm \eqref{mnorm}.

At this point we continue by defining $ \Phi^{\times} $ to be the collection of all \textit{antilinear} functionals over $ \Phi $ bounded w.r.t. \eqref{mnorm}. This definition originates in the theory of distributions (the only difference is that distributions are linear). In particular, for every locally integrable function $ f : \mathbb{R}^{q} \rightarrow \mathbb{R} $ there is a corresponding functional $  F_{f} \in \Phi^{\times} $ of the form
\begin{equation}\label{Ffphi}
F_{f}(\varphi) = \int_{\mathbb{R}^{q}} \overline{\varphi(x)} f(x) ~ d^{q}x
\end{equation}
for all $ \varphi \in \Phi $. As for the elements in $ \Phi^{\times} $ that are not of this form, we will make extensive use of $ \overline{\delta}_{a} \in \Phi^{\times} $ defined by $ \overline{\delta}_{a}(\varphi) = \overline{\varphi(a)} $ with a parameter $ a \in \mathbb{R}^{q} $. This is nothing but the $ q $-dimensional antilinear variation on Dirac delta. We warn the reader that we will commonly abuse notation by adopting the symbol $ \delta^{q}(x-a) $ from the expression
\begin{equation}\label{overlinedeltaa}
\overline{\delta}_{a}(\varphi) \equiv \int_{\mathbb{R}^{q}} \overline{\varphi(x)} \delta^{q}(x-a) ~ d^{q}x
\end{equation}
where the mathematical language still falls short in reflecting Dirac's genius. Moreover, we shall write $  F_{f} \equiv \vert f \rangle $ and $ \overline{\delta}_{a} \equiv \vert a \rangle $ to accommodate for the usual notation $ F_{f}(\varphi) \equiv \langle \varphi \vert f \rangle $ and $ \overline{\delta}_{a}(\varphi) \equiv \langle \varphi \vert a \rangle $. We will generally talk about the elements of $ \Phi^{\times} $ as \textit{(right) states}.

Following the standard recipe described in \cite{Madrid2005}, one introduces another rigged Hilbert space $ \Phi \subset \mathcal{F} \subset \Phi^{\overline{\times}} $ analogical to the one given above, where $ \Phi^{\overline{\times}} $ is defined to be the collection of all \textit{linear} functionals over $ \Phi $ bounded w.r.t. \eqref{mnorm}. In mathematical terms, functionals in $ \Phi^{\overline{\times}} $ are tempered distributions. It comes as little surprise that there is a one-to-one correspondence between $ \Phi^{\times} $ and $ \Phi^{\overline{\times}} $. And again, given a locally integrable function $ f : \mathbb{R}^{q} \rightarrow \mathbb{R} $, there is a corresponding functional $  \overline{F_{f}} \in \Phi^{\overline{\times}} $ of the form
\begin{equation}\label{key}
\overline{F_{f}}(\varphi) = \int_{\mathbb{R}^{q}} \varphi(x) \overline{f(x)} ~ d^{q}x
\end{equation}
As in the preceding case, we also introduce $ \delta_{a} \in \Phi^{\overline{\times}}  $ acting as $ \delta_{a}(\varphi) = \varphi(a) $ which can be put to the same integral form $ \delta_{a}(\varphi) \equiv \int_{\mathbb{R}^{q}} \varphi(x) \delta^{q}(x-a) ~ d^{q}x $ with the aid of Dirac delta function $ \delta^{q}(x-a) $. We shall write $ \overline{F_{f}} \equiv \langle f \vert $ and $ \delta_{a} \equiv \langle a \vert $ to accommodate for the notation $ \overline{F_{f}}(\varphi) \equiv \langle f \vert \varphi \rangle $ and $ \delta_{a}(\varphi) \equiv \langle a \vert \varphi \rangle $. We will generally talk about the elements of $ \Phi^{\overline{\times}} $ as \textit{(left) states}.

For any \textit{self-adjoint} operator $ \hat{A} : \Phi \rightarrow \Phi $, we define the corresponding operator $ \hat{A} : \Phi^{\times} \rightarrow \Phi^{\times} $ by $ (\hat{A} F)(\varphi) = F(\hat{A} \varphi) $ for all $ F \in \Phi^{\times} $ and $ \varphi \in \Phi $. We will say that the functional $ \alpha \in \Phi^{\times} $ is a \textit{(right) eigenstate} of $ \hat{A} $ with eigenvalue $ A \in \mathbb{R} $ if $ \hat{A} \alpha = A \alpha $. This is to be understood as an equality on $ \Phi^{\times} $, i.e., $ (\hat{A} \alpha)(\varphi) = (A \alpha)(\varphi) $ for all $ \varphi \in \Phi $. In Dirac's notation, the above definitions are written as $ \langle \varphi \vert \hat{A} \vert F \rangle = \langle \hat{A} \varphi \vert F \rangle $ and $ \hat{A} \vert \alpha \rangle = A \vert \alpha \rangle $. The definitions for operators on $ \Phi^{\overline{\times}} $ and their \textit{(left) eigenstates} are completely analogical, yielding $ \langle F \vert \hat{A} \vert \varphi \rangle = \langle F \vert \hat{A} \varphi \rangle $ and $ \langle \alpha \vert \hat{A} = A \langle \alpha \vert $. Note that they well apply to $ \hat{y}_{A}, \hat{y}_{A+q} $ since both of them are self-adjoint on $ \Phi $.

\vspace{\baselineskip}

We may finally define the eigenvalues $ \alpha_{A}, \beta_{A} \in \mathbb{R} $ as well as the eigenstates $ \vert \alpha \rangle_{\mathtt{c}}, \vert \beta \rangle_{\mathtt{m}} \in \Phi^{\times} $ of the observables $ \hat{y}_{A}, \hat{y}_{A+q} $, respectively, via the equations
\begin{equation}\label{yAalpha}
\hat{y}_{A} \vert \alpha \rangle_{\mathtt{c}} = \alpha_{A} \vert \alpha \rangle_{\mathtt{c}}
\end{equation}
\begin{equation}\label{yAqalpha}
\hat{y}_{A+q} \vert \beta \rangle_{\mathtt{m}} = \beta_{A} \vert \beta \rangle_{\mathtt{m}}
\end{equation}
We use the subindices in $ \vert \alpha \rangle_{\mathtt{c}} $ and $ \vert \beta \rangle_{\mathtt{m}} $ to signify that they are the coordinate and momentum eigenstates (of the coordinate and momentum observables with respect to the canonical basis), respectively. By virtue of the one-to-one correspondence between $ \Phi^{\times} $ and $ \Phi^{\overline{\times}} $, one also gets the functionals $ _{\mathtt{c}} \langle \alpha \vert, ~ _{\mathtt{m}} \langle \beta \vert \in \Phi^{\overline{\times}} $. These indeed are the left eigenstates of $ \hat{y}_{A}, \hat{y}_{A+q} $ in $ \Phi^{\overline{\times}} $. Upon solving the equations \eqref{yAalpha} and \eqref{yAqalpha} one finds the eigenstates to be
\begin{equation}\label{cm}
\begin{aligned}
\vert \alpha \rangle_{\mathtt{c}} &= \overline{\delta}_{\alpha}\\
\vert \beta \rangle_{\mathtt{m}} &= \sum_{\alpha} (2\pi)^{-q/2} ~ e^{ i \beta_{A} \alpha_{A} } ~ \overline{\delta}_{\alpha}
\end{aligned}
\end{equation}
where we employ the symbolic summation
\begin{equation}\label{sumint}
\sum_{\alpha} \equiv \int_{\mathbb{R}} \prod_{A = 1}^{q} d \alpha_{A}
\end{equation}
We will use this notation throughout the paper for its brevity.

\vspace{\baselineskip}

The above construction provides us with the powerful tool of Dirac's formalism. We shall only shortly explain how it can be understood. The key feature is that every state comes with an integral form obtained formally by writing
\begin{equation}\label{FoverlineG}
\begin{aligned}
F = \sum_{\alpha} \overline{\delta}_{\alpha} \left( \delta_{\alpha} \bullet F \right) \\
\bar{G} = \sum_{\alpha} \left( \bar{G} \bullet \overline{\delta}_{\alpha} \right) \delta_{\alpha}
\end{aligned}
\end{equation}
for any $ F \in \Phi^{\times} $ and $ \bar{G} \in \Phi^{\overline{\times}} $. We shall assume that this form always exists without going into mathematical details. The equations \eqref{FoverlineG} can be viewed as defining relations for $ \delta_{\alpha} \bullet F  $ and $ \bar{G} \bullet \overline{\delta}_{\alpha} $. Note that these have meaning only inside the integral (one could treat them rigorously as measures). For instance, comparing the first row of \eqref{FoverlineG} to \eqref{Ffphi}, we immediately get $ \delta_{\alpha} \bullet F_{f} = f(\alpha) $ where it is understood that $ \alpha \in \mathbb{R}^{q} $ is the designated integration variable. Similarly, from \eqref{overlinedeltaa} it follows $ \delta_{\alpha} \bullet \overline{\delta}_{\gamma} = \delta^{q}(\alpha-\gamma) $ where $ \alpha \in \mathbb{R}^{q} $ is the integration variable and $ \gamma \in \mathbb{R}^{q} $ is a parameter. Yet another way of writing the same is
\begin{equation}\label{key}
\begin{aligned}
\hat{\mathbf{1}} = \sum_{\alpha} \overline{\delta}_{\alpha} \delta_{\alpha} \bullet \\
\hat{\mathbf{1}} = \sum_{\alpha} \bullet \overline{\delta}_{\alpha} \delta_{\alpha}
\end{aligned}
\end{equation}
Note that the symbol $ \hat{\mathbf{1}} $ has two different meanings here: in the first line, it stands for the identity operator on $ \Phi^{\times} $ while in the second line it stands for the identity operator on $ \Phi^{\overline{\times}} $. When we apply the described philosophy on the eigenstates \eqref{cm}, we arrive at the notorious relations
\begin{equation}\label{formfieldmom}
\begin{aligned}
_{\mathtt{c}} \langle \alpha \vert \gamma \rangle _{\mathtt{c}} &= \delta^{q}(\alpha - \gamma) \\
_{\mathtt{c}} \langle \alpha \vert \beta \rangle_{\mathtt{m}} &= (2\pi)^{-q/2} ~ e^{ i \beta_{A} \alpha_{A} } \\
\hat{\mathbf{1}} &= \sum_{\alpha} \vert \alpha \rangle_{\mathtt{c}} ~ _{\mathtt{c}} \langle \alpha \vert \\
\hat{\mathbf{1}} &= \sum_{\beta} \vert \beta \rangle _{\mathtt{m}} ~ _{\mathtt{m}} \langle \beta \vert
\end{aligned}
\end{equation}
here reproduced in Dirac's notation. For the lack of a better name, the objects $ _{\mathtt{c}} \langle \alpha \vert \gamma \rangle _{\mathtt{c}} $ and $ _{\mathtt{c}} \langle \alpha \vert \beta \rangle_{\mathtt{m}} $ will be called \textit{wavefunctions}, although they are \textit{not} functions. Instead, they must be understood as abstract expressions of the form $ \delta_{\alpha} \bullet F  $ which have a specific effect on the integral that they happen to be part of.

\section{Algebraic Preliminaries}
In this section we briefly review some basic tools and results from linear algebra that will be indispensable for our work. The first topic of interest shall be the singular value decomposition and Moore-Penrose pseudoinverse, whose treatment will be based on the reference \cite{Golub2012}. Then we shortly remind the elementary properties of symplectic matrices which can be found in \cite{Gosson2006}.

\begin{theorem}
	Let $ A \in \mathbb{R}^{m \times n}  $ be an $ m \times n $ matrix with $ m \geq n $. Then there exist orthogonal matrices $ U \in \mathbb{R}^{m \times m} $ and $ V \in \mathbb{R}^{n \times n} $ and a matrix $ \Sigma = \begin{pmatrix}
	\text{diag}(\sigma_{1}, ..., \sigma_{n}) \\
	0
	\end{pmatrix} \in \mathbb{R}^{m \times n} $ with $ \sigma_{1} \geq \sigma_{2} \geq ... \geq \sigma_{n} \geq 0 $, such that
	\begin{equation}\label{key}
	A = U \Sigma V^{T}
	\end{equation}
	The numbers $ \sigma_{1}, ..., \sigma_{n} $ are called \textit{singular values} of $ A $. If $ \sigma_{r} > 0 $ is the smallest nonzero singular value, then the matrix $ A $ has rank $ r $.
\end{theorem}

The assumption $ m \geq n $ is used here for simplicity, the singular value decomposition exists for any matrix. Nevertheless, since we will be interested in square matrices, the given formulation is more than sufficient. We also remark that the decomposition is not unique---only the matrix $ \Sigma $ is uniquely determined by $ A $.

We will take advantage of the notation $ U = ( U_{1} ~ U_{2} ) $ and $ V = ( V_{1} ~ V_{2} ) $ with $ U_{1} \in \mathbb{R}^{m \times r} $, $ U_{2} \in \mathbb{R}^{m \times n-r} $, $ V_{1} \in \mathbb{R}^{n \times r} $ and $ V_{2} \in \mathbb{R}^{n \times n-r} $, and further denote $ \Sigma_{r} = \text{diag}(\sigma_{1}, ..., \sigma_{r}) \in \mathbb{R}^{r \times r} $. Then one can write
\begin{equation}\label{key}
A = \begin{pmatrix} U_{1} & U_{2} \end{pmatrix} \begin{pmatrix}
\Sigma_{r} & 0 \\
0 & 0
\end{pmatrix} \begin{pmatrix}
V_{1}^{T} \\
V_{2}^{T}
\end{pmatrix} = U_{1} \Sigma_{r} V_{1}^{T}
\end{equation}

Next we define the Moore-Penrose pseudoinverse as follows:
\begin{definition}\label{MP}
	Let $ A \in \mathbb{R}^{m \times n} $ be a matrix and $ A = U \Sigma V^{T} = U_{1} \Sigma_{r} V_{1}^{T} $ its (narrowed) singular value decomposition with $ \Sigma_{r} = \text{diag} ( \sigma_{1}, ..., \sigma_{r} ) \in \mathbb{R}^{r \times r} $. Then the matrix $ A^{+} = V \Sigma^{+} U^{T} = V_{1} \Sigma_{r}^{+} U_{1}^{T} $ with $ \Sigma^{+} = \begin{pmatrix}
	\Sigma_{r}^{+} & 0 \\
	0 & 0
	\end{pmatrix} \in \mathbb{R}^{n \times m} $ and $ \Sigma_{r}^{+} = \text{diag} ( \sigma_{1}^{-1}, ..., \sigma_{r}^{-1}) \in \mathbb{R}^{r \times r} $ is called the \textit{Moore-Penrose pseudoinverse} of $ A $.
\end{definition}

\begin{theorem}\label{Penrose}
	(Penrose Equations). The Moore-Penrose pseudoinverse $ A^{+} $ of $ A $ is the only solution of the matrix equations
	\begin{equation}\label{key}
	\begin{aligned}
	&\text{(i)} ~ A A^{+} A = A & \qquad &\text{(iii)} ~ ( A A^{+} )^{T} = A A^{+} \\
	&\text{(ii)} ~ A^{+} A A^{+} = A^{+} & \qquad &\text{(iv)} ~ ( A^{+} A )^{T} = A^{+} A
	\end{aligned}
	\end{equation}
\end{theorem}

Let us now remind the fundamental spaces associated to a matrix $ A $, together with their basic properties, and provide the corresponding projectors in terms of the singular value decomposition.
\begin{definition}\label{rowcolnull}
	We define the following fundamental spaces:
	\begin{enumerate}[{(1)}]
		\item $ \mathcal{R}(A) = \{ y ~ \vert ~ \exists x \in \mathbb{R}^{n} : y = A x \} \subset \mathbb{R}^{m}  $ is the \textit{range} or \textit{column space}.
		\item $ \mathcal{R}(A^{T}) = \{ z ~ \vert ~ \exists y \in \mathbb{R}^{n} : z = A^{T} y \} \subset \mathbb{R}^{n}  $ is the \textit{row space}.
		\item $ \mathcal{N}(A) = \{ x ~ \vert ~ A x = 0 \} \subset \mathbb{R}^{n}  $ is the \textit{null space}.
	\end{enumerate}
\end{definition}

\begin{theorem}
	The following relations hold:
	\begin{enumerate}[{(i)}]
		\item $ \mathcal{R}(A)^{\perp} = \mathcal{N}(A^{T}) $, therefore $ \mathbb{R}^{m} = \mathcal{R}(A) \oplus \mathcal{N}(A^{T}) $.
		\item $ \mathcal{R}(A^{T})^{\perp} = \mathcal{N}(A) $, therefore $ \mathbb{R}^{n} = \mathcal{R}(A^{T}) \oplus \mathcal{N}(A) $.
	\end{enumerate}
\end{theorem}

\begin{theorem}
	The projectors to the spaces of Definition \ref{rowcolnull} are given by
	\begin{equation}\label{key}
	\begin{aligned}
	& P_{\mathcal{R}(A)} = AA^{+} & \qquad & P_{\mathcal{R}(A^{T})} = A^{+} A \\
	& P_{\mathcal{N}(A^{T})} = \mathbf{1} - AA^{+} & \qquad & P_{\mathcal{N}(A)} = \mathbf{1} - A^{+} A
	\end{aligned}
	\end{equation}
	Alternatively, using the singular value decomposition,
	\begin{equation}\label{key}
	\begin{aligned}
	& P_{\mathcal{R}(A)} = U_{1} U_{1}^{T} & \qquad & P_{\mathcal{R}(A^{T})} = V_{1} V_{1}^{T} \\
	& P_{\mathcal{N}(A^{T})} = U_{2} U_{2}^{T} & \qquad & P_{\mathcal{N}(A)} = V_{2} V_{2}^{T}
	\end{aligned}
	\end{equation}
\end{theorem}

The Moore-Penrose pseudoinverse is of great importance to us because it can be readily used to write an explicit solution to a general linear set of equations. Consider the matrix problem
\begin{equation}\label{Axb}
A x = b
\end{equation}
with $ A \in \mathbb{R}^{m \times n} $ a matrix, $ x \in \mathbb{R}^{n} $ and $ b \in \mathbb{R}^{m} $. The equation is consistent, and therefore has a solution for $ x $, only if $ b \in \mathcal{R}(A) $. This condition (sometimes also referred to as \textit{constraint}) can be equivalently expressed as
\begin{equation}\label{Axbconstr}
U_{2} U_{2}^{T} b = 0
\end{equation}
where we employed the narrowed singular value decomposition $ A = U_{1} \Sigma_{r} V_{1}^{T} $ and projected the equation \eqref{Axb} onto $ \mathcal{N}(A^{T}) $ via $ P_{\mathcal{N}(A^{T})} = U_{2} U_{2}^{T} $. There are two special cases in which the constraint is satisfied automatically, namely $ b = 0 $ and $ \mathcal{R}(A) = \mathbb{R}^{m} $.

If the constraint \eqref{Axbconstr} holds, there is a family of solutions for $ x $ of the form
\begin{equation}\label{key}
x = A^{+} b + V_{2} c
\end{equation}
where $ c \in \mathbb{R}^{s} $ is an arbitrary vector of dimension $ s \equiv n - r $. We remark that this is exactly the solution of the linear least squares problem $ A x \approx b $ which comes around by projecting the right-hand side of \eqref{Axb} onto $ \mathcal{R}(A) $ and thus solving the equation $ A x = A A^{+} b $ rather than \eqref{Axb}. One can see that this is equivalent to simply ignoring the constraint \eqref{Axbconstr}.

\vspace{\baselineskip}

In the rest, we shall briefly recall the definition of a symplectic matrix and review its elementary properties.
\begin{definition}
	A \textit{symplectic matrix} $ W $ is a real $ 2Q \times 2Q $ matrix satisfying
	\begin{equation}\label{key}
	W^{T} \sigma W = \sigma
	\end{equation}
\end{definition}
with
\begin{equation}\label{key}
\sigma = \begin{pmatrix}
\phantom{-} 0 & \mathbf{1} \\
-\mathbf{1} & 0
\end{pmatrix}
\end{equation}

\begin{theorem}\label{sympl}
	Let us denote
	\begin{equation}\label{key}
	W = \begin{pmatrix}
	E & F \\
	G & H
	\end{pmatrix}
	\end{equation}
	where $ E, F, G, H $ are real $ Q \times Q $ matrices. Then the following conditions are equivalent:
	\begin{enumerate}
		\item The matrix $ W $ is symplectic.
		\item $ E^{T}G $, $ F^{T} H $ are symmetric and $ E^{T} H - G^{T} F = \mathbf{1} $
		\item $ E F^{T} $, $ G H^{T} $ are symmetric and $ E H^{T} - F G^{T} = \mathbf{1} $
	\end{enumerate}
\end{theorem}

It follows from condition 2. that the inverse of a symplectic matrix $ W $ is
\begin{equation}\label{W-1}
W^{-1} = \begin{pmatrix}
\phantom{-}H^{T} & -F^{T} \\
-G^{T} & \phantom{-}E^{T}
\end{pmatrix}
\end{equation}

\section{The Symplectic Transformation of Observables}

Let us consider the following problem. In the phase space $ \mathcal{P} $ we have the canonical (symplectic) basis $ \{ e_{J} \}_{J = 1}^{2q} $ and the quantum observables $ \hat{y}_{A} $ and $ y_{A+q} $ measure the coordinates and momenta with respect to this basis. Now suppose we are given a new symplectic basis $ \{ g_{I} \}_{I = 1}^{2q} $ of $ \mathcal{P} $ defined by a linear transformation
\begin{equation}\label{key}
e_{J} = g_{I} W_{IJ}
\end{equation}
with a $ 2q \times 2q $ \textit{symplectic} matrix of real coefficients $ W_{IJ} $. We will denote
\begin{equation}\label{W}
W = \begin{pmatrix}
E & F \\
G & H
\end{pmatrix}
\end{equation}
where $ E, F, G, H $ are real $ q \times q $ matrices. A vector $ y \in \mathcal{P} $ may be written as $ y = y_{J} e_{J} = w_{I} g_{I} $, with a linear (symplectic) coordinate transformation of the form
\begin{equation}\label{key}
w_{I} = W_{IJ} ~ y_{J}
\end{equation}
Our aim is to transform the observables on $ \Phi $ accordingly:
\begin{equation}\label{wsympl}
\hat{w}_{I} = W_{IJ} ~ \hat{y}_{J}
\end{equation}
One expects that the canonical commutation relations will not be touched by the symplectic transformation \eqref{wsympl}, since, in their nature, they are nothing but a quantum variation on the symplectic form $ \omega $. To check this explicitly, first observe that (as immediately follows from Theorem \ref{sympl}) $ W $ is symplectic $ \Leftrightarrow  W^{T} $ is symplectic, i.e., it holds $ W \sigma W^{T} = \sigma $. Then we easily find
\begin{equation}\label{key}
[ \hat{w}_{I} , \hat{w}_{J} ] = W_{IK} W_{JL} [ \hat{y}_{K} , \hat{y}_{L} ] = W_{IK} W_{JL} i \sigma_{KL} \hat{\mathbf{1}} = i (W \sigma W^{T})_{IJ} \hat{\mathbf{1}} = i \sigma_{IJ} \hat{\mathbf{1}}
\end{equation}
which is the exact same form as \eqref{ccr}.

Before we go on to look for the new eigenstates, we should make sure that we still have the right rigged Hilbert space they can live on. This is indeed the case, because the definition \eqref{Phidef} of $ \Phi $ is clearly invariant with respect to a linear transformation like \eqref{wsympl}. We can therefore use the same rigged Hilbert spaces $ \Phi \subset \mathcal{F} \subset \Phi^{\times} $ and $ \Phi \subset \mathcal{F} \subset \Phi^{\overline{\times}} $ in dealing with the new observables $ \hat{w}_{I} $.

\subsection{The Transformation of Coordinates}

Let us define new eigenstates $ \vert \omega \rangle_{\dot{\mathtt{c}}} \in \Phi^{\times} $ and eigenvalues $ \omega_{A} \in \mathbb{R} $ via
\begin{equation}\label{omegadef}
\hat{w}_{A} \vert \omega \rangle_{\dot{\mathtt{c}}} = \omega_{A} \vert \omega \rangle_{\dot{\mathtt{c}}}
\end{equation}
with a new set of observables
\begin{equation}\label{key}
\hat{w}_{A} = W_{AJ} ~ \hat{y}_{J}
\end{equation}
In the general case which interests us, the matrix $ W_{AJ} $ will mix coordinates and momenta, and $ \vert \omega \rangle_{\dot{\mathtt{c}}} $ will turn out to be different from the eigenstates $ \vert \alpha \rangle_{\mathtt{c}} $ of $ \hat{y}_{A} $.

Our main aim is to find $ \vert \omega \rangle_{\dot{\mathtt{c}}} $ in the coordinate eigenstate basis
\begin{equation}\label{key}
\vert \omega \rangle_{\dot{\mathtt{c}}} = \sum_{\alpha} \vert \alpha \rangle_{\mathtt{c}} ~ _{\mathtt{c}} \langle \alpha \vert \omega \rangle_{\dot{\mathtt{c}}}
\end{equation}
In the following, we elaborate on that. First let us remind that according to the definition \eqref{yder}, it holds
\begin{equation}\label{momonbasis}
_{\mathtt{c}} \langle \alpha \vert \hat{y}_{A+q} \vert \omega \rangle_{\dot{\mathtt{c}}} = - i \frac{\partial}{\partial \alpha_{A}} ~ _{\mathtt{c}} \langle \alpha \vert \omega \rangle_{\dot{\mathtt{c}}}
\end{equation}
and the defining relation \eqref{omegadef} for $ \vert \omega \rangle_{\dot{\mathtt{c}}} $ can be written as
\begin{equation}\label{eigenvalinbasis}
~ _{\mathtt{c}} \langle \alpha \vert \hat{w}_{A} \vert \omega \rangle_{\dot{\mathtt{c}}} = \omega_{A} ~ _{\mathtt{c}} \langle \alpha \vert \omega \rangle_{\dot{\mathtt{c}}}
\end{equation}
It follows from \eqref{momonbasis} and \eqref{eigenvalinbasis} that
\begin{equation}\label{key}
E_{AB} \alpha_{B} ~ _{\mathtt{c}} \langle \alpha \vert \omega \rangle_{\dot{\mathtt{c}}} - F_{AB} ~ i \frac{\partial}{\partial \alpha_{B}} ~ _{\mathtt{c}} \langle \alpha \vert \omega \rangle_{\dot{\mathtt{c}}} = \omega_{A} ~ _{\mathtt{c}} \langle \alpha \vert \omega \rangle_{\dot{\mathtt{c}}}
\end{equation}
In case that $ _{\mathtt{c}} \langle \alpha \vert \omega \rangle_{\dot{\mathtt{c}}} \neq 0 $, we divide by it and get (in matrix notation)
\begin{equation}\label{EFalom}
E \alpha + F a = \omega
\end{equation}
with
\begin{equation}\label{an}
a_{B} = - i \frac{\partial}{\partial \alpha_{B}} \ln ~ _{\mathtt{c}} \langle \alpha \vert \omega \rangle_{\dot{\mathtt{c}}}
\end{equation}
In the following we analyze the equation \eqref{EFalom} in an attempt to find a solution for $ ~ _{\mathtt{c}} \langle \alpha \vert \omega \rangle_{\dot{\mathtt{c}}} $. There may be couples $ (\alpha, \omega) $ for which no solution exists, then it must hold $ ~ _{\mathtt{c}} \langle \alpha \vert \omega \rangle_{\dot{\mathtt{c}}} = 0 $. Of course, the product may be zero even if there is a nonzero solution for it; one should be therefore careful about one's conclusions.

The first step in analyzing \eqref{EFalom} is to realize that it can be split into two fundamentally different parts. Since $ F $ is a general $ q \times q $ real matrix, it is not necessarily invertible. Upon employing the singular value decomposition
\begin{equation}\label{key}
F = U \Sigma V^{T} = U_{1} \Sigma_{r} V_{1}^{T}
\end{equation}
we can split \eqref{EFalom} into two equations
\begin{equation}\label{EFalomsplit}
\begin{aligned}
U_{1}^{T} E \alpha + \Sigma_{r} V_{1}^{T} a &= U_{1}^{T} \omega\\
U_{2}^{T} E \alpha &= U_{2}^{T} \omega
\end{aligned}
\end{equation}
They are obtained by multiplying \eqref{EFalom} from left by $ U_{1}^{T} $ and $ U_{2}^{T} $, respectively, and using $ U_{1}^{T} U_{1} = \mathbf{1} $ and $ U_{2}^{T} U_{1} = 0 $. The first row represents a linear set of $ r = \text{rank}(F) $ equations for $ a $, the second is an additional condition of dimension $ s = q - r $.

We first look closer at the first row of \eqref{EFalomsplit}. In general, it does not fix $ a $ uniquely, since it only contains $ r $ equations for a vector of dimension $ q $. Nevertheless, it can be used to fix at least a part of $ a $. Upon multiplying by $ V_{1} \Sigma_{r}^{-1} $, we get
\begin{equation}\label{key}
V_{1} V_{1}^{T} a = V_{1} \Sigma_{r}^{-1} U_{1}^{T} \omega - V_{1} \Sigma_{r}^{-1} U_{1}^{T} E \alpha
\end{equation}
where we recognize the pseudoinverse $ F^{+} = V_{1} \Sigma_{r}^{-1} U_{1}^{T} $ from Definition~\ref{MP} and rewrite
\begin{equation}\label{key}
V_{1} V_{1}^{T} a = F^{+}( \omega - E \alpha )
\end{equation}
The vector $ a $ is projected by $ P_{\mathcal{R}(F^{T})} = F^{+} F = V_{1} V_{1}^{T} $ to the row space $ \mathcal{R}(F^{T}) $ of $ F $. It follows that
\begin{equation}\label{key}
a = F^{+}( \omega - E \alpha ) + V_{2} \kappa
\end{equation}
with $ \kappa \in \mathbb{R}^{s} $ free. Now that we have expressed $ a $, we may use \eqref{an} to gain information about $ _{\mathtt{c}} \langle \alpha \vert \omega \rangle_{\dot{\mathtt{c}}} $. The starting point is
\begin{equation}\label{alphapde}
 ( F^{+}\omega + V_{2} \kappa ) - F^{+} E \alpha = - i \frac{\partial}{\partial \alpha} \ln ~ _{\mathtt{c}} \langle \alpha \vert \omega \rangle_{\dot{\mathtt{c}}}
\end{equation}

We would like to integrate along the following rule: given that
\begin{equation}\label{key}
\frac{\partial}{\partial \zeta} f(\zeta) = A \zeta + b
\end{equation}
with $ A $ symmetric, one finds the primitive function to be of the form
\begin{equation}\label{key}
f(\zeta) = \frac{1}{2} \zeta^{T} M \zeta + \zeta^{T} b + C
\end{equation}
with $ M $ satisfying $ A = \frac{1}{2} ( M + M^{T} ) $. Without loss of generality, one can choose $ A = M $.

The problem with \eqref{alphapde} is that $ F^{+} E $ is generally not symmetric. That is why we first need to prepare grounds for the integration. We start by splitting
\begin{equation}\label{alphabetagamma}
\alpha = \beta + \gamma
\end{equation}
with $ \beta \in \mathcal{R}(F^{T}) $ and $ \gamma \in \mathcal{N}(F) $. Since the two subspaces are orthogonal, the splitting is unique; the respective parts are $ \beta = F^{+} F \alpha = V_{1} V_{1}^{T} \alpha $ and $ \gamma = (\mathbf{1} - F^{+} F ) \alpha = V_{2} V_{2}^{T} \alpha $. It turns out one can comfortably integrate over $ \beta $. One prepares
\begin{equation}\label{key}
- i \frac{\partial}{\partial \alpha} \ln ~ _{\mathtt{c}} \langle \alpha \vert \omega \rangle_{\dot{\mathtt{c}}} = - i \frac{\partial}{\partial \beta} \ln ~ _{\mathtt{c}} \langle \beta + \gamma \vert \omega \rangle_{\dot{\mathtt{c}}}
\end{equation}
It follows from \eqref{alphapde}
\begin{equation}\label{key}
- i \frac{\partial}{\partial \beta} \ln ~ _{\mathtt{c}} \langle \beta + \gamma \vert \omega \rangle_{\dot{\mathtt{c}}} = ( F^{+}\omega + V_{2} \kappa - F^{+} E \gamma ) - F^{+} E \beta
\end{equation}
One can substitute $ \beta = F^{+} F \beta $. Then, using the symmetry of $ E F^{T} $ from Theorem \ref{sympl} and the symmetry of $ F^{+} F = ( F^{+} F )^{T} = F^{T} F^{+T} $ given by Theorem \ref{Penrose}, one finds that $ F^{+} E F^{+} F = F^{+} E F^{T} F^{+T} $ is symmetric. Thus one obtains
\begin{equation}\label{key}
- i \frac{\partial}{\partial \beta} \ln ~ _{\mathtt{c}} \langle \beta + \gamma \vert \omega \rangle_{\dot{\mathtt{c}}} = ( F^{+}\omega + V_{2} \kappa - F^{+} E \gamma ) - F^{+} E F^{T} F^{+T} \beta
\end{equation}
which can be integrated as
\begin{equation}\label{ilnbeta}
- i \ln ~ _{\mathtt{c}} \langle \beta + \gamma \vert \omega \rangle_{\dot{\mathtt{c}}} = G(\gamma) + \beta^{T} ( F^{+}\omega + V_{2} \kappa - F^{+} E \gamma ) - \frac{1}{2} \beta^{T} F^{+} E F^{T} F^{+T} \beta
\end{equation}

To get the full expression for $ i \ln ~ _{\mathtt{c}} \langle \beta + \gamma \vert \omega \rangle_{\dot{\mathtt{c}}} $, we also need to look at the derivative w.r.t. $ \gamma $. It is analogical:
\begin{equation}\label{gammapdf}
- i \frac{\partial}{\partial \gamma} \ln ~ _{\mathtt{c}} \langle \beta + \gamma \vert \omega \rangle_{\dot{\mathtt{c}}} = ( F^{+}\omega + V_{2} \kappa - F^{+} E \beta ) - F^{+} E \gamma
\end{equation}
Interestingly, one can again get the matrix $ F^{+} E $ into a symmetric form, only with a different trick. By definition $ F \gamma = 0 $, and it follows that $ F^{+T} \gamma = 0 $. Then one can simply symmetrize $ F^{+} E \gamma = (F^{+} E + E^{T} F^{+T} ) \gamma \equiv M \gamma $. But there is a problem: the quadratic term which would result from an integration of \eqref{gammapdf} turns out to be $ \gamma^{T} M \gamma = 0 $. This cannot be the primitive function, because it could only give rise to a zero derivative. That is, a contradiction appears unless
\begin{equation}\label{gammaconstr}
F^{+} E \gamma = 0
\end{equation}
We conclude that for \eqref{alphapde} to be true, $ \gamma $ must satisfy this condition. Then
\begin{equation}\label{ilngamma}
\begin{aligned}
- i \ln ~ _{\mathtt{c}} \langle \beta + \gamma \vert \omega \rangle_{\dot{\mathtt{c}}} &= B(\beta) + \gamma^{T} ( F^{+}\omega + V_{2} \kappa - F^{+} E \beta ) = \\
&= B(\beta) + \gamma^{T} ( F^{+}\omega + V_{2} \kappa )
\end{aligned}
\end{equation}
The simplification in the second row follows from $ \gamma^{T} F^{+} = 0 $. We could also discard the first term in the bracket, but let us keep it. The next step is to match the prescriptions \eqref{ilnbeta} and \eqref{ilngamma}. At first sight, we find a discrepancy because the term $ - \beta^{T} F^{+} E \gamma $ in \eqref{ilnbeta} cannot occur in \eqref{ilngamma}. However, the term is conveniently annihilated by \eqref{gammaconstr}. The result of the matching is
\begin{equation}\label{ilnbetagamma}
- i \ln ~ _{\mathtt{c}} \langle \beta + \gamma \vert \omega \rangle_{\dot{\mathtt{c}}} = C + (\beta + \gamma )^{T} ( F^{+}\omega + V_{2} \kappa ) - \frac{1}{2} \beta^{T} F^{+} E F^{T} F^{+T} \beta
\end{equation}
where $ C $ is a complex constant. We remark that $ C $ may still depend on the parameters of the problem, like $ \omega $, $ E $ and $ F $. Next, we can return to the formulation with $ \alpha $ by writing $ \beta^{T} F^{+} E F^{T} F^{+T} \beta = \alpha^{T} F^{+} E F^{T} F^{+T} \alpha $, and get
\begin{equation}\label{ilnalpha}
- i \ln ~ _{\mathtt{c}} \langle \alpha \vert \omega \rangle_{\dot{\mathtt{c}}} = C + \alpha^{T} ( F^{+}\omega + V_{2} \kappa ) - \frac{1}{2} \alpha^{T} F^{+} E F^{T} F^{+T} \alpha
\end{equation}
Finally, we summarize our findings as
\begin{equation}\label{ao}
~ _{\mathtt{c}} \langle \alpha \vert \omega \rangle_{\dot{\mathtt{c}}} = e^{i C} e^{i \left( - \frac{1}{2} \alpha^{T} F^{+} E F^{T} F^{+T} \alpha + \alpha^{T} ( F^{+}\omega + V_{2} \kappa ) \right) }
\end{equation}
\begin{equation}\label{gammaconstr2}
F^{+} E V_{2} V_{2}^{T} \alpha = 0
\end{equation}
The latter is an equivalent form of the condition \eqref{gammaconstr}.

Let us continue with an analysis of the additional conditions. So far, we have obtained two of them: besides \eqref{gammaconstr2}, we also have the original condition given in the second row of \eqref{EFalomsplit}, which is equivalently written as
\begin{equation}\label{constr}
U_{2} U_{2}^{T} E \alpha = U_{2} U_{2}^{T} \omega
\end{equation}
We will show how these two conditions limit the possible values of the involved variables and what are the consequences. Let us start by making the following observation.

\begin{observation} \label{bg}
	Provided that $ \alpha = \beta + \gamma $ with $ \beta \in \mathcal{R}(F^{T}) $ and $ \gamma \in \mathcal{N}(F) $, the couple \eqref{gammaconstr2} and \eqref{constr} is equivalent to
	\begin{equation}\label{Engamma2}
	E \gamma = U_{2} U_{2}^{T} \omega
	\end{equation}
\end{observation}
\begin{proof}
	The splitting---note that it was used before in \eqref{alphabetagamma}---allows us to study
	\begin{equation}\label{Enalphan}
	E \alpha = E \beta + E \gamma
	\end{equation}
	Upon arranging
	\begin{equation}\label{key}
	E \beta = E F^{+} F \beta = E F^{T} F^{+T} \beta = F E^{T} F^{+T} \beta
	\end{equation}
	it is found that $ E \beta \in  \mathcal{R}(E) \cap \mathcal{R}(F) $.
	
	Assume that \eqref{gammaconstr2} and \eqref{constr} hold. Then $ E \gamma \in \mathcal{R}(E) \cap \mathcal{N}(F^{T}) $, as implied by \eqref{gammaconstr2}. Compare \eqref{Enalphan} to
	\begin{equation}\label{key}
	E \alpha = F F^{+} E \alpha + U_{2} U_{2}^{T}  E \alpha
	\end{equation}
	where $ E \alpha $ is subjected to a standard splitting into two mutually orthogonal parts $ F F^{+} E \alpha \in \mathcal{R}(F) $ and $ U_{2} U_{2}^{T}  E \alpha \in \mathcal{N}(F^{T}) $. It follows
	\begin{equation}\label{Enbetan}
	E \beta = F F^{+} E \alpha
	\end{equation}
	\begin{equation}\label{Engamma}
	E \gamma = U_{2} U_{2}^{T}  E \alpha
	\end{equation}
	We note that using \eqref{Enalphan} and \eqref{gammaconstr2}, one can also rewrite \eqref{Enbetan} as $ E \beta = F F^{+} E \beta $. This equation is not needed for the proof anyway, we are giving it only for completeness. By plugging \eqref{constr} into \eqref{Engamma}, one obtains \eqref{Engamma2}.
	
	The other direction can be proven analogously. Assume that \eqref{Engamma2} holds. Since $ E $ maps onto $ \mathcal{R}(E) $ and $ U_{2} U_{2}^{T} $ is a projector onto $ \mathcal{N}(F^{T}) $, a direct consequence is $ E \gamma \in \mathcal{R}(E) \cap \mathcal{N}(F^{T}) $. The equations \eqref{Enbetan} and \eqref{Engamma} follow, and it is a matter of arrangement to obtain both \eqref{gammaconstr2} and \eqref{constr}.
\end{proof}

The form of the conditions established in Observation \ref{bg} is much simpler and can be used to solve for $ \alpha $. Interestingly, it turns out that $ \beta $ is not constrained by these at all. Therefore we only need to solve for $ \gamma $. First of all, looking at \eqref{Engamma2}, we see that it requires $ \omega $ to satisfy
\begin{equation}\label{omegaconstr}
U_{2} U_{2}^{T} \omega \in \mathcal{R}(E)
\end{equation}
This result seems unexpected, because it constraints the parameters $ \omega $ of the problem. If $ \omega $ violated \eqref{omegaconstr}, we would have to conclude that $ _{\mathtt{c}} \langle \alpha \vert \omega \rangle_{\dot{\mathtt{c}}} = 0 $ for all $ \alpha $, i.e., $ \vert \omega \rangle_{\dot{\mathtt{c}}} = 0 $, which would be quite strange. However, it turns out that this case does not occur:

\begin{observation}\label{NETNFT}
	The condition \eqref{omegaconstr} is always satisfied.
\end{observation}
\begin{proof}
	The observation follows from the regularity of $ W $. We can prove it easily from the transpose of the identity in condition 3. of Theorem \ref{sympl}, i.e., $ H E^{T} - G F^{T} = \mathbf{1} $. Assuming $ \iota \in \mathcal{N}(E^{T}) \cap \mathcal{N}(F^{T}) $, we get $ 0 = H E^{T}\iota  - G F^{T}\iota = \iota $, and therefore $ \mathcal{N}(E^{T}) \cap \mathcal{N}(F^{T}) = \{ 0 \} $. Since by definition $ U_{2} U_{2}^{T} \omega \in \mathcal{N}(F^{T}) $, we know that it cannot be in $ \mathcal{N}(E^{T}) $ unless it is 0. In any case, it follows that $ U_{2} U_{2}^{T} \omega \in \mathcal{R}(E) $.
\end{proof}

We continue with the following observation.
\begin{observation}
	The condition \eqref{Engamma2} with $ \gamma \in \mathcal{N}(F) $ is equivalent to the prescription
	\begin{equation}\label{gammafinal}
	\gamma = E^{+} U_{2} U_{2}^{T} \omega + \mu
	\end{equation}
	with $ \mu \in \mathcal{N}(E) \cap \mathcal{N}(F) $.
\end{observation}
\begin{proof}
	The equation \eqref{Engamma2} can be solved straightforwardly as $ \gamma = E^{+} U_{2} U_{2}^{T} \omega + \mu $ with $ \mu \in \mathcal{N}(E) $. We only need to make sure that $ \gamma \in \mathcal{N}(F) $, or in other words, $ F^{+} F E^{+} U_{2} U_{2}^{T} \omega + F^{+} F \mu = 0 $. We arrange
	\begin{equation}\label{key}
	\begin{aligned}
	F^{+} F E^{+} U_{2} U_{2}^{T} \omega &= F^{T} F^{+T} E^{+} U_{2} U_{2}^{T} \omega = \\ &= F^{T} \left( E F^{T} \right)^{+} U_{2} U_{2}^{T} \omega = \\ &= F^{T} \left( F E^{T} \right)^{+} U_{2} U_{2}^{T} \omega = \\ &= F^{T} E^{+T} F^{+} U_{2} U_{2}^{T} \omega
	\end{aligned}
	\end{equation}
	and since $ F^{+} U_{2} = 0 $, we see that this term vanishes; that is, $ E^{+} U_{2} U_{2}^{T} \omega \in \mathcal{N}(F) \cap \mathcal{R}(E^{T}) $. The condition $ \gamma \in \mathcal{N}(F) $ therefore turns into $ F^{+} F \mu = 0 $, i.e., $ \mu \in \mathcal{N}(F) $. The other direction of the implication can be proven analogically.
\end{proof}

Further simplification is achieved with:

\begin{observation}\label{nzero}
	It holds $ \mathcal{N}(E) \cap \mathcal{N}(F) = \{ 0 \} $.
\end{observation}
\begin{proof}
	Suppose $ \mu \in \mathcal{N}(E) \cap \mathcal{N}(F) $, i.e., $ F \mu = E \mu = 0 $. Then, using the identities in condition 2. and 3. of Theorem \ref{sympl}, one gets $ E^{T} H \mu = \mu $ and $ - F^{T} G \mu = \mu $. It follows that $ \mu \in \mathcal{R}(E^{T}) $ as well as $ \mu \in \mathcal{R}(F^{T}) $. Since $ \mathcal{N}(E) \perp \mathcal{R}(E^{T}) $, $ \mathcal{N}(F) \perp \mathcal{R}(F^{T}) $, each one of these is enough to conclude that $ \mu = 0 $.
\end{proof}

We have now completed the calculations by giving all additional conditions a compact form: we found that there is no condition on $ \omega $, while $ \alpha $ is constrained by the simple prescription
\begin{equation}\label{gammaconstrfin}
\gamma = E^{+} U_{2} U_{2}^{T} \omega
\end{equation}
gained from \eqref{gammafinal} together with Observation \ref{nzero}. The equation \eqref{gammaconstrfin} needs to be incorporated into a final expression for $ \vert \omega \rangle_{\dot{\mathtt{c}}} $ now. Our educated guess is that it may be done with a Dirac delta function. In particular, we conjecture that $ _{\mathtt{c}} \langle \alpha \vert \omega \rangle_{\dot{\mathtt{c}}} $ should be proportional to
\begin{equation}\label{key}
\delta^{q} \left( \gamma - E^{+} U_{2} U_{2}^{T} \omega \right)
\end{equation}
Such term would then appear in
\begin{equation}\label{key}
\vert \omega \rangle_{\dot{\mathtt{c}}} = \sum_{\alpha} \vert \alpha \rangle_{\mathtt{c}} ~_{\mathtt{c}} \langle \alpha \vert \omega \rangle_{\dot{\mathtt{c}}}
\end{equation}
The present summation over $ \alpha $ can be split into two, one over $ \beta \in \mathcal{R}(F^{T}) $ and the other over $ \gamma \in \mathcal{N}(F) $,
\begin{equation}\label{key}
\vert \omega \rangle_{\dot{\mathtt{c}}} = \sum_{\beta} \sum_{\gamma} \vert \beta + \gamma \rangle_{\mathtt{c}} ~ _{\mathtt{c}}\langle \beta + \gamma \vert \omega \rangle_{\dot{\mathtt{c}}}
\end{equation}
The conditions $ \beta \in \mathcal{R}(F^{T}) $, $ \gamma \in \mathcal{N}(F) $ shall be implemented implicitly as
\begin{equation}\label{key}
\sum_{\beta} \equiv \sum_{\beta} \delta^{s}(V_{2}^{T}\beta), \qquad \sum_{\gamma} \equiv \sum_{\gamma} \delta^{r}(V_{1}^{T}\gamma)
\end{equation}
In other words, it should be understood that the delta functions are there although we will not write them. Their arguments enforce that the integration is taken over the subspaces stated: recall that it holds $ \beta = V_{1} V_{1}^{T} \beta + V_{2} V_{2}^{T} \beta $ with $ V_{1} V_{1}^{T} \beta \in \mathcal{R}(F^{T}) $ and $ V_{2} V_{2}^{T} \beta \in \mathcal{N}(F) $, so if the latter is zero, we get $ \beta \in \mathcal{R}(F^{T}) $. Similarly for $ \gamma \in \mathcal{N}(F) $. Altogether, plugging in \eqref{ao}, we expect a result of the form
\begin{equation}\label{key}
\begin{aligned}
\vert \omega \rangle_{\dot{\mathtt{c}}} &= \sum_{\beta} \sum_{\gamma} \vert \alpha \rangle_{\mathtt{c}} ~ e^{i C} e^{i \left( - \frac{1}{2} \alpha^{T} F^{+} E F^{T} F^{+T} \alpha + \alpha^{T} ( F^{+}\omega + V_{2} \kappa ) \right) } ~ \delta^{q} \left( \gamma - E^{+} U_{2} U_{2}^{T} \omega \right)
\end{aligned}
\end{equation}
where we use a shorthand $ \alpha \equiv \beta + \gamma $, or
\begin{equation}\label{omegaresult}
\vert \omega \rangle_{\dot{\mathtt{c}}} = \sum_{\beta} \vert \alpha \rangle_{\mathtt{c}} ~ e^{i C} e^{i \left( - \frac{1}{2} \alpha^{T} F^{+} E F^{T} F^{+T} \alpha + \alpha^{T} ( F^{+}\omega + V_{2} \kappa ) \right) }
\end{equation}
where we have updated to $ \alpha \equiv \beta + E^{+} U_{2} U_{2}^{T} \omega $, still with $ \beta \in \mathcal{R}(F^{T}) $. Let us recall that $ C $ is a complex constant w.r.t. $ \alpha $, which may nevertheless depend on $ \omega $ or $ W $. The real part of $ C $ is irrelevant, the imaginary part shall serve as normalization. The vector $ \kappa $ of dimension $ s = q - \text{rank} ( F ) $ is arbitrary.

We can simplify
\begin{equation}\label{omegaresult2}
\vert \omega \rangle_{\dot{\mathtt{c}}} = \sum_{\beta} \vert \beta + E^{+} U_{2} U_{2}^{T} \omega \rangle_{\mathtt{c}} ~ e^{i C} e^{i \left(- \frac{1}{2} \beta^{T} F^{+} E F^{T} F^{+T} \beta - \beta^{T} F^{+}\omega + \omega^{T} U_{2} U_{2}^{T} E^{+T} V_{2} \kappa \right) }
\end{equation}
and while modifying the real part of $ C $ (with only a slight abuse of notation), we can throw away the constant term in the exponent, getting
\begin{equation}\label{omegaresult3}
\vert \omega \rangle_{\dot{\mathtt{c}}} = \sum_{\beta} \vert \beta + E^{+} U_{2} U_{2}^{T} \omega \rangle_{\mathtt{c}} ~ e^{i C} e^{i \left(- \frac{1}{2} \beta^{T} F^{+} E F^{T} F^{+T} \beta + \beta^{T} F^{+}\omega \right) }
\end{equation}
We have therefore got rid of all the arbitrariness in the expression.

\vspace{\baselineskip}

Let us summarize the conclusions of this paragraph. Considering the eigenvalue problem $ \hat{w} \vert \omega \rangle_{\dot{\mathtt{c}}} = \omega \vert \omega \rangle_{\dot{\mathtt{c}}} $ with the set of observables $ \hat{w} = W  \hat{y} $ (given by a symplectic transformation with a matrix $ W $), we found that the eigenstate $ \vert \omega \rangle_{\dot{\mathtt{c}}} $ satisfies
\begin{equation}\label{omegaresultfin}
\begin{aligned}
\vert \omega \rangle_{\dot{\mathtt{c}}} &= \sum_{\beta} \vert \beta + E^{+} U_{2} U_{2}^{T} \omega \rangle_{\mathtt{c}} ~ e^{i C} e^{i \left( -\frac{1}{2} \beta^{T} F^{+} E F^{T} F^{+T} \beta + \beta^{T} F^{+}\omega \right) }
\\ & \text{ with } \beta \in \mathcal{R}(F^{T})
\end{aligned}
\end{equation}
Eventually, let us express the integration over $ \beta $ in a more explicit way. We lay down $ \beta = V \chi = V_{1} \xi + V_{2} \zeta $ and perform the substitution
\begin{equation}\label{key}
\sum_{\beta} \delta^{s}(V_{2}^{T}\beta) = \sum_{\chi} \lvert \det V \rvert ~ \delta^{s}(\zeta) = \lvert \det V \rvert \sum_{\xi} = \sum_{\xi}
\end{equation}
where $ \chi \in \mathbb{R}^{q} $ while $ \xi \in \mathbb{R}^{r} $ with $ r = \text{rank}(F) $. Note that since $ V $ is an orthogonal matrix, $ \lvert \det V \rvert = 1 $. We may then rewrite \eqref{omegaresultfin} as
\begin{equation}\label{omegaresultfin2}
\vert \omega \rangle_{\dot{\mathtt{c}}} = \sum_{\xi} \vert V_{1} \xi + E^{+} U_{2} U_{2}^{T} \omega \rangle_{\mathtt{c}} ~ e^{i C} e^{i \left( -\frac{1}{2} \xi^{T} V_{1}^{T} F^{+} E F^{T} F^{+T} V_{1} \xi + \xi^{T} V_{1}^{T} F^{+}\omega \right) }
\end{equation}

\subsection{Wavefunctions and the Resolution of Identity}
One expects that the wavefunctions $ _{\dot{\mathtt{c}}}\langle \varrho \vert \omega \rangle_{\dot{\mathtt{c}}} $ given by the eigenstates $ \vert \varrho \rangle_{\dot{\mathtt{c}}} , \vert \omega \rangle_{\dot{\mathtt{c}}} \in  \Phi^{\times} $ of $ \hat{w}_{A} $ will produce Dirac delta functions, as in the case of the original observables. We will check this now. First we prepare grounds by performing a linear substitution in the Dirac delta function. We have
\begin{equation}\label{key}
\delta^{q}(x) = (2 \pi)^{-q} \int_{\mathbb{R}} d^{q}p ~ e^{i p^{T} x}
\end{equation}
with matrix notation in the exponent. Then, for a $ q \times q $ regular matrix $ V $, we find
\begin{equation}\label{key}
\begin{aligned}
\delta^{q}(V \tau) &= (2 \pi)^{-q} \int_{\mathbb{R}} d^{q}p ~ e^{i p^{T} V \tau} = (2 \pi)^{-q} \lvert \det V^{T} \rvert^{-1} \int_{\mathbb{R}} d^{q}\nu ~ e^{i \nu^{T} \tau} = \lvert \det V^{T} \rvert^{-1} \delta^{q}(\tau)
\end{aligned}
\end{equation}
We shall use this substitution with an orthogonal matrix $ V $, for which one has simply
\begin{equation}\label{key}
\delta^{q}(V \tau) = \delta^{q}(\tau)
\end{equation}
Next, considering a vector
\begin{equation}\label{key}
\psi = \begin{pmatrix}
\theta \\
\phi
\end{pmatrix}
\end{equation}
with $ \theta \in \mathbb{R}^{r} $, $ \phi \in \mathbb{R}^{s} $, one can arrange
\begin{equation}\label{deltadecomp}
\begin{aligned}
\delta^{q}( V_{1} \theta + V_{2} \phi ) = \delta^{q}( V \psi ) = \delta^{q}( \psi ) = \delta^{r}( \theta ) \delta^{s}( \phi )
\end{aligned}
\end{equation}
In the following, we shall use \eqref{deltadecomp} with $ \theta = \zeta - \xi $, $ \phi = V_{2}^{T} E^{+} U_{2} U_{2}^{T} (\varrho - \omega) $. The latter shall come around thanks to the fact that $ E^{+} U_{2} U_{2}^{T} (\varrho - \omega) = V_{2} V_{2}^{T} E^{+} U_{2} U_{2}^{T} (\varrho - \omega) $. With these preparations, we can move on to compute the product of \eqref{omegaresultfin2} with
\begin{equation}\label{key}
\vert \varrho \rangle_{\dot{\mathtt{c}}} = \sum_{\zeta} \vert V_{1} \zeta + E^{+} U_{2} U_{2}^{T} \varrho \rangle_{\mathtt{c}} ~ e^{i B} e^{i \left( -\frac{1}{2} \zeta^{T} V_{1}^{T} F^{+} E F^{T} F^{+T} V_{1} \zeta + \zeta^{T} V_{1}^{T} F^{+} \varrho \right) }
\end{equation}
We write
\begin{equation}\label{varrhoomprod}
\begin{aligned}
_{\dot{\mathtt{c}}}\langle \varrho \vert \omega \rangle_{\dot{\mathtt{c}}} &= \sum_{\zeta} \sum_{\xi} ~ _{\mathtt{c}}\langle V_{1} \zeta + E^{+}U_{2}U_{2}^{T} \varrho \vert V_{1} \xi + E^{+} U_{2} U_{2}^{T} \omega  \rangle_{\mathtt{c}} ~  e^{i \left( C - \overline{B} \right)} \\ & \qquad e^{-i \left(- \frac{1}{2} \zeta^{T} V_{1}^{T} F^{+} E F^{T} F^{+T} V_{1} \zeta + \zeta^{T} V_{1}^{T} F^{+}\varrho \right) } ~ e^{i \left( -\frac{1}{2} \xi^{T} V_{1}^{T} F^{+} E F^{T} F^{+T} V_{1} \xi + \xi^{T} V_{1}^{T} F^{+}\omega \right) } = \\
&= \sum_{\zeta} \sum_{\xi} \delta^{r}(\zeta - \xi)  \delta^{s}(  V_{2}^{T} E^{+} U_{2} U_{2}^{T} (\varrho - \omega) ) e^{i \left( C - \overline{B} \right)} \\ & \qquad e^{-i \left( -\frac{1}{2} \zeta^{T} V_{1}^{T} F^{+} E F^{T} F^{+T} V_{1} \zeta + \zeta^{T} V_{1}^{T} F^{+}\varrho \right) } ~ e^{i \left( -\frac{1}{2} \xi^{T} V_{1}^{T} F^{+} E F^{T} F^{+T} V_{1} \xi + \xi^{T} V_{1}^{T} F^{+}\omega \right) } = \\
&= \sum_{\xi} \delta^{s}( V_{2}^{T} E^{+} U_{2} U_{2}^{T} (\varrho - \omega) ) e^{i \left( C - \overline{B} \right)} ~ e^{i  \xi^{T} V_{1}^{T} F^{+} (\omega - \varrho)  } = \\
&= \delta^{s}( V_{2}^{T} E^{+} U_{2} U_{2}^{T} (\varrho - \omega) ) e^{i \left( C - \overline{B} \right)} (2\pi)^{r} \delta^{r}( V_{1}^{T} F^{+} ( \varrho - \omega) )
\end{aligned}
\end{equation}
Let us see what the two conditions provided in the Dirac deltas say about $ \eta \equiv \varrho - \omega $. The wavefunction shall be nonzero if and only if both $ V_{2}^{T} E^{+} U_{2} U_{2}^{T} \eta = 0 $ and $ V_{1}^{T} F^{+} \eta = 0 $ hold. These equations are equivalent to $  E^{+} ( \mathbf{1} - F F^{+} ) \eta = 0 $ and $ F^{+} \eta = 0 $. The first equation demands that $ ( \mathbf{1} - F F^{+} ) \eta \in \mathcal{N}(E^{T}) $. At the same time, $ ( \mathbf{1} - F F^{+} ) \eta \in \mathcal{N}(F^{T}) $ because $ ( \mathbf{1} - F F^{+} ) $ is a projector to $ \mathcal{N}(F^{T}) $. From these we get (see the proof of Observation \ref{NETNFT}) that $( \mathbf{1} - F F^{+} ) \eta = 0 $. Together with the second equation, which implies $ F F^{+} \eta = 0 $, this tells us that $ \eta = 0 $. We therefore find that up to a constant, which can be removed by a suitable choice of $ C $ and $ B $, the wavefunction is
\begin{equation}\label{inprovarrhoomega}
_{\dot{\mathtt{c}}}\langle \varrho \vert \omega \rangle_{\dot{\mathtt{c}}} = \delta^{q}(\varrho - \omega)
\end{equation}
as expected.

Once we have established the wavefunctions \eqref{inprovarrhoomega}, we have at our disposal the resolution of identity;
\begin{equation}\label{key}
\vert \varrho \rangle_{\dot{\mathtt{c}}} = \sum_{\omega} \delta^{q}(\omega - \varrho) \vert \omega \rangle_{\dot{\mathtt{c}}} = \sum_{\omega} \vert \omega \rangle_{\dot{\mathtt{c}}} \phantom{.}_{\dot{\mathtt{c}}}\langle \omega \vert \varrho \rangle_{\dot{\mathtt{c}}} 
\end{equation}
from where we get
\begin{equation}\label{key}
\mathbf{1} = \sum_{\omega} \vert \omega \rangle_{\dot{\mathtt{c}}} \phantom{.}_{\dot{\mathtt{c}}}\langle \omega \vert
\end{equation}
However, we must not forget that this formula only holds if the eigenstates $ \vert \omega \rangle_{\dot{\mathtt{c}}} $ are correctly normalized according to \eqref{inprovarrhoomega}.

\subsection{The Transformation of Momenta}

So far we have computed the eigenstates $ \vert \omega \rangle_{\dot{\mathtt{c}}} $ of the coordinate observables $ \hat{w}_{A}  $. We now turn our attention to the analogical problem
\begin{equation}\label{acuteomegadef}
\hat{w}_{A+q} \vert \omega \rangle_{\dot{\mathtt{m}}} = \omega_{A} \vert \omega \rangle_{\dot{\mathtt{m}}}
\end{equation}
for momentum observables
\begin{equation}\label{key}
\hat{w}_{A+q} = W_{A+q J} ~ \hat{y}_{J}
\end{equation}
The solution too shall be analogical. From \eqref{yder}, one gets
\begin{equation}\label{momonbasis2}
~_{\mathtt{c}}\langle \alpha \vert \hat{y}_{A+q} \vert \omega \rangle_{\dot{\mathtt{m}}} = - i \frac{\partial}{\partial \alpha_{A}} ~_{\mathtt{c}}\langle \alpha \vert \omega \rangle_{\dot{\mathtt{m}}}
\end{equation}
and the defining relation \eqref{acuteomegadef} for $ \vert \omega \rangle_{\dot{\mathtt{m}}} $ can be rewritten as
\begin{equation}\label{eigenvalinbasis2}
~_{\mathtt{c}}\langle \alpha \vert \hat{w}_{A+q} \vert \omega \rangle_{\dot{\mathtt{m}}} = \omega_{A} ~_{\mathtt{c}}\langle \alpha \vert \omega \rangle_{\dot{\mathtt{m}}}
\end{equation}
It follows from \eqref{momonbasis2} and \eqref{eigenvalinbasis2} that
\begin{equation}\label{key}
G_{AB} \alpha_{B} ~_{\mathtt{c}}\langle \alpha \vert \omega \rangle_{\dot{\mathtt{m}}} - H_{AB} ~ i \frac{\partial}{\partial \alpha_{B}}  ~_{\mathtt{c}}\langle \alpha \vert \omega \rangle_{\dot{\mathtt{m}}} = \omega_{A} ~_{\mathtt{c}}\langle \alpha \vert \omega \rangle_{\dot{\mathtt{m}}}
\end{equation}
In case that $ _{\mathtt{c}}\langle \alpha \vert \omega \rangle_{\dot{\mathtt{m}}} \neq 0 $, we divide by it and get (in matrix notation)
\begin{equation}\label{EFalom2}
G \alpha + H \acute{a} = \omega
\end{equation}
while denoting
\begin{equation}\label{an2}
\acute{a}_{B} = - i \frac{\partial}{\partial \alpha_{B}} \ln ~_{\mathtt{c}}\langle \alpha \vert \omega \rangle_{\dot{\mathtt{m}}}
\end{equation}

The equation \eqref{EFalom2} then can be treated in complete analogy with the previous paragraph. This time it is important that the matrix $ G H^{T} $ is symmetric, and we may take advantage of the following:
\begin{observation}
	It holds $ \mathcal{N}(H^{T}) \cap \mathcal{N}(G^{T}) = \{ 0 \} $.
\end{observation}
\begin{proof}
	We recall the identity in condition 3. of Theorem \ref{sympl} which states $ E H^{T} - F G^{T} = \mathbf{1} $, and act with both right and left hand side on $ \nu \in \mathcal{N}(H^{T}) \cap \mathcal{N}(G^{T}) $, getting $ \nu = 0 $.
\end{proof}

\begin{observation}
	It holds $ \mathcal{N}(H) \cap \mathcal{N}(G) = \{ 0 \} $.
\end{observation}
\begin{proof}
	We recall the identity in condition 2. of Theorem \ref{sympl} which states $ E^{T} H - G^{T} F = \mathbf{1} $, and act with both right and left hand side on $ \mu \in \mathcal{N}(H) \cap \mathcal{N}(G) $, getting $ - G^{T} F \mu = \mu $. It follows that $ \mu \in \mathcal{R}(G^{T}) $, but since $ \mathcal{R}(G^{T}) \perp \mathcal{N}(G) $, we get $ \mu = 0 $.
\end{proof}

With these in place, the solution of \eqref{EFalom2} comes around in the same form as \eqref{omegaresultfin},
\begin{equation}\label{omegaacuteresultfin}
\begin{aligned}
\vert \omega \rangle_{\dot{\mathtt{m}}} &= \sum_{\beta} \vert \beta + G^{+} \acute{U}_{2} \acute{U}_{2}^{T} \omega \rangle_{\mathtt{c}} ~ e^{i K} e^{i \left( -\frac{1}{2} \beta^{T} H^{+} G H^{T} H^{+T} \beta + \beta^{T} H^{+} \omega \right) }
\\ & \text{ with } \beta \in \mathcal{R}(H^{T})
\end{aligned}
\end{equation}
For brevity, we denote $ \acute{U}_{2} \equiv U_{2}(H) $, $ \acute{V}_{1} \equiv V_{1}(H) $, etc. The alternative form with an explicit integration over $ \xi \in \mathbb{R}^{\acute{r}} $ where $ \acute{r} \equiv \text{rank}(H) $ is
\begin{equation}\label{omegaacuteresultfin2}
\vert \omega \rangle_{\dot{\mathtt{m}}} = \sum_{\xi} \vert \acute{V}_{1} \xi + G^{+} \acute{U}_{2} \acute{U}_{2}^{T} \omega \rangle_{\mathtt{c}} ~ e^{i K} e^{i \left( -\frac{1}{2} \xi^{T} \acute{V}_{1}^{T} H^{+} G H^{T} H^{+T} \acute{V}_{1} \xi + \xi^{T} \acute{V}_{1}^{T} H^{+} \omega \right) }
\end{equation}

The computation of wavefunctions from the preceding paragraph carries over to this case, too. Thus we get
\begin{equation}\label{varrhoomprodacute}
_{\dot{\mathtt{m}}}\langle \varrho \vert \omega \rangle_{\dot{\mathtt{m}}} = \delta^{\acute{s}}( \acute{V}_{2}^{T} G^{+} \acute{U}_{2} \acute{U}_{2}^{T} (\varrho - \omega) ) e^{i \left( K - \overline{L} \right)} (2\pi)^{\acute{r}} \delta^{\acute{r}}( \acute{V}_{1}^{T} H^{+} ( \varrho - \omega) )
\end{equation}
with $ \acute{r} \equiv \text{rank}(H) $ and $ \acute{s} \equiv q - \acute{r} $. It follows that there are constants $ K $ and $ L $ serving as normalization of $ \vert \omega \rangle_{\dot{\mathtt{m}}} $ and $ \vert \varrho \rangle_{\dot{\mathtt{m}}} $, respectively, such that
\begin{equation}\label{inprovarrhoomegaacute}
_{\dot{\mathtt{m}}}\langle \varrho \vert  \omega \rangle_{\dot{\mathtt{m}}} = \delta^{q}(\varrho - \omega)
\end{equation}
With this normalization in place, one has the resolution of identity
\begin{equation}\label{key}
\mathbf{1} = \sum_{\omega} \vert \omega \rangle_{\dot{\mathtt{m}}} ~_{\dot{\mathtt{m}}}\langle \omega \vert
\end{equation}

\section{Examples}
In this section we offer a handful of special cases of the $ \hat{w} = W \hat{y} $ transformation defined by a symplectic matrix $ W $ and test our results on them. We start with two very prominent choices and add another one to illustrate the differences in their behavior.

\begin{example}\label{E1}
	The first prominent case is
	\begin{equation}\label{W1}
	W = \begin{pmatrix}
	O & 0 \\
	0 & O
	\end{pmatrix}
	\end{equation}
	where $ O \in \mathbb{R}^{q \times q} $ is an \textit{orthogonal} matrix. One easily checks that $ W $ is symplectic. The transformation does not mix coordinates and momenta, which makes it exceptionally simple. Let us apply our analysis to \eqref{W1}. We plug $ E = O $, $ F = 0 $ into \eqref{omegaresultfin} and observe that because $ \mathcal{R}(F^{T}) = \{ 0 \} $, it holds $ \beta = 0 $. Also, $ U_{2} U_{2}^{T} = \mathbf{1} $. The integration is therefore trivial and we are left only with
	\begin{equation}\label{key}
	\vert \omega \rangle_{\dot{\mathtt{c}}} = \vert O^{T} \omega \rangle_{\mathtt{c}} ~ e^{i C}
	\end{equation}
	We choose normalization by fixing $ C = 0 $ to end up with the result
	\begin{equation}\label{omgw1}
	\vert \omega \rangle_{\dot{\mathtt{c}}} = \vert O^{T} \omega \rangle_{\mathtt{c}}
	\end{equation}
	Since
	\begin{equation}\label{key}
	\hat{w} \vert \omega \rangle_{\dot{\mathtt{c}}} = O \hat{x} \vert O^{T} \omega \rangle_{\mathtt{c}} = O O^{T} \omega \vert O^{T} \omega \rangle_{\mathtt{c}} = \omega \vert \omega \rangle_{\dot{\mathtt{c}}}
	\end{equation}
	the result is obviously correct.
	
	Next, let us look at the momenta. We plug $ G = 0 $, $ H = O $ into \eqref{omegaacuteresultfin}, which implies $ \beta = \alpha $ and $ \acute{U}_{2} \acute{U}_{2}^{T} = 0 $. In result, we get
	\begin{equation}\label{key}
	\vert \omega \rangle_{\dot{\mathtt{m}}} = \sum_{\alpha} \vert \alpha \rangle_{\mathtt{c}} ~ e^{i K} e^{i  \alpha^{T} O^{T} \omega }
	\end{equation}
	In this case we opt for the normalization $ e^{i K} = (2\pi)^{-q/2} $, obtaining
	\begin{equation}\label{key}
	\vert \omega \rangle_{\dot{\mathtt{m}}} = \sum_{\alpha} \vert \alpha \rangle_{\mathtt{c}} ~ (2\pi)^{-q/2} ~ e^{i \omega^{T} O \alpha }
	\end{equation}
	In the trivial case $ O = \mathbf{1} $, one reproduces the transformation between the coordinate eigenstates $ \vert \alpha \rangle_{\mathtt{c}} $ and the momentum eigenstates $ \vert \omega \rangle_{\dot{\mathtt{m}}} \equiv \vert \beta \rangle_{\mathtt{m}} $ in the form
	\begin{equation}\label{acutebeta}
	\vert \beta \rangle_{\mathtt{m}} = \sum_{\alpha} \vert \alpha \rangle_{\mathtt{c}} ~ (2\pi)^{-q/2} ~ e^{i \beta^{T} \alpha}
	\end{equation}
	which is a direct consequence of \eqref{formfieldmom}.
\end{example}

\vspace{\baselineskip}

\begin{example}\label{E2}
	The second prominent case occurs when
	\begin{equation}\label{W2}
	W = \begin{pmatrix}
	0 & O \\
	- O & 0
	\end{pmatrix}
	\end{equation}
	again with $ O \in \mathbb{R}^{q \times q} $ orthogonal. This as well is a symplectic matrix, and we see that the resulting transformation effectively exchanges coordinates with (a mixture of) momenta, and vice versa. Up to this exchange, one expects to obtain similar results to those in Example \ref{E1}.
	
	Upon plugging $ E = 0 $, $ F = O $ into \eqref{omegaresultfin}, we get $ \mathcal{R}(F^{T}) = \mathbb{R}^{q} $ and so $ \beta = \alpha $. Then
	\begin{equation}\label{key}
	\vert \omega \rangle_{\dot{\mathtt{c}}} = \sum_{\alpha} \vert \alpha \rangle_{\mathtt{c}} ~ e^{i C} e^{i \alpha^{T} O^{T}\omega }
	\end{equation}
	We shall again normalize with the choice $ e^{i C} = (2\pi)^{-q/2} $, obtaining
	\begin{equation}\label{omegamom}
	\vert \omega \rangle_{\dot{\mathtt{c}}} = \sum_{\alpha} \vert \alpha \rangle_{\mathtt{c}} ~ (2\pi)^{-q/2} ~ e^{ i \omega^{T} O \alpha }
	\end{equation}
	For the typical case $ O = \mathbf{1} $, each coordinate value is simply replaced with its corresponding momentum, and one may identify $ \vert \omega \rangle_{\dot{\mathtt{c}}} = \vert \omega \rangle_{\mathtt{m}} $. The resulting form
	\begin{equation}\label{key}
	\vert \omega \rangle_{\mathtt{m}} = \sum_{\alpha} \vert \alpha \rangle_{\mathtt{c}} ~ (2\pi)^{-q/2} ~ e^{i \omega^{T} \alpha}
	\end{equation}
	is identical to \eqref{acutebeta} as one expects.
	
	As for momenta, deploying $ G = -O $ and $ H = 0 $ on \eqref{omegaacuteresultfin} renders $ \beta = 0 $ and $ \acute{U}_{2} \acute{U}_{2}^{T} = \mathbf{1} $. This, together with the choice $ K = 0 $, leaves us with the result
	\begin{equation}\label{key}
	\vert \omega \rangle_{\dot{\mathtt{m}}} = \vert - O^{T} \omega \rangle_{\mathtt{c}}
	\end{equation}
	analogical to \eqref{omgw1}. The reader can easily verify that it is correct. Here the choice $ O = \mathbf{1} $ yields $ \vert \omega \rangle_{\dot{\mathtt{m}}} = - \vert \omega \rangle_{\mathtt{c}} $.
\end{example}

\vspace{\baselineskip}

\begin{example}\label{E3}
	Consider the matrix
	\begin{equation}\label{W3}
	W = \frac{1}{\sqrt{2}} \begin{pmatrix}
	\phantom{-} \mathbf{1} & \mathbf{1} \\
	- \mathbf{1} & \mathbf{1}
	\end{pmatrix}
	\end{equation}
	It is symplectic, and for the first time introduces a non-trivial mixing of coordinates and momenta. We designed the matrix to have the simplest regular blocks $ E = F = \frac{1}{\sqrt{2}} \mathbf{1} $ possible. We have $ \mathcal{R}(F^{T}) = \mathbb{R}^{q} $. It follows that $ \beta = \alpha $ as in Example \ref{E2}, but the quadratic term in the integrand does not vanish this time, instead one obtains
	\begin{equation}\label{key}
	\vert \omega \rangle_{\dot{\mathtt{c}}} = \sum_{\alpha} \vert \alpha \rangle_{\mathtt{c}} ~ e^{i C} e^{i \left( -\frac{1}{2} \alpha^{T} \alpha + \sqrt{2} \alpha^{T} \omega \right) }
	\end{equation}
	If one uses the identity
	\begin{equation}\label{key}
	\delta^{q}(k x) = \lvert k \rvert^{-q} ~ \delta^{q}(x) 
	\end{equation}
	for $ k \in \mathbb{R} $, $ x \in \mathbb{R}^{q} $ within \eqref{varrhoomprod}, one finds that $ e^{i C} = ( \sqrt{2} \pi )^{-q/2} $ is the correct normalization of $ \vert \omega \rangle_{\dot{\mathtt{c}}} $.
	
	For momenta the situation is alike. It holds $ \beta = \alpha $ and $ \acute{U}_{2} \acute{U}_{2}^{T} = 0 $, and we find
	\begin{equation}\label{key}
	\vert \omega \rangle_{\dot{\mathtt{m}}} = \sum_{\alpha} \vert \alpha \rangle_{\mathtt{c}} ~ e^{i K} e^{i \left( \frac{1}{2} \alpha^{T} \alpha + \sqrt{2} \alpha^{T} \omega \right) }
	\end{equation}
	i.e., the only difference between the coordinate and momentum eigenstates is the sign of the quadratic term in the exponent. From \eqref{varrhoomprodacute}, we get the same normalization $ e^{i K} = ( \sqrt{2} \pi )^{-q/2} $ as above.
	
	Eventually, let us establish the wavefunctions of the momentum eigenstates. We shall use the formula
	\begin{equation}\label{iAJ}
	\int_{\mathbb{R}^{n}} e^{i \left( \frac{1}{2} x^{T} A x + J^{T} x \right)} d^{n} x = (2 \pi i)^{n/2} \left( \det A \right)^{-1/2} e^{-i \frac{1}{2} J^{T} A^{-1} J}
	\end{equation}
	for $ A \in \mathbb{R}^{n \times n} $ a real, symmetric, invertible matrix and $ J \in \mathbb{R}^{n} $. This is a multi-dimensional version of a formula which can be found in Supplement I of \cite{Chaichian2001}. Strictly speaking, the integral in \eqref{iAJ} is divergent; the formula only holds in the sense of regularization which is done by including the real term $ - \eta x^{T} x $ in the exponent and taking $ \lim_{\eta \rightarrow 0^{+}} $. With this help, we are able to compute
	\begin{equation}\label{key}
	\begin{aligned}
	_{\dot{\mathtt{m}}} \langle \varrho \vert \omega \rangle_{\dot{\mathtt{c}}} &= \sum_{\gamma} \sum_{\alpha} ~_{\mathtt{c}} \langle \gamma \vert \alpha \rangle_{\mathtt{c}} ~ ( \sqrt{2} \pi )^{q} ~ e^{-i \left( \frac{1}{2} \gamma^{T} \gamma + \sqrt{2} \gamma^{T} \varrho \right) } e^{i \left(- \frac{1}{2} \alpha^{T} \alpha + \sqrt{2} \alpha^{T} \omega \right) } = \\
	&= \sum_{\alpha} ( \sqrt{2} \pi )^{q} ~ e^{-i \left( \frac{1}{2} \alpha^{T} \alpha + \sqrt{2} \alpha^{T} \varrho \right) } e^{i \left(-\frac{1}{2} \alpha^{T} \alpha + \sqrt{2} \alpha^{T} \omega \right) } = \\
	&= \sum_{\alpha} ( \sqrt{2} \pi )^{q} ~ e^{i \left(- \alpha^{T} \alpha + \sqrt{2} \alpha^{T} (\omega - \varrho ) \right) } = \\
	&= (-i2 \pi^{3})^{q/2} ~ e^{ i\frac{1}{2} (\omega - \varrho )^{T} (\omega - \varrho )}
	\end{aligned}
	\end{equation}
	One can see that in the special case $ \omega = \varrho $, the product $ _{\dot{\mathtt{m}}} \langle \omega \vert \omega \rangle_{\dot{\mathtt{c}}} $ is constant. This behavior is quite different from that of the original coordinate and momentum eigenstates, where $ _{\mathtt{c}} \langle \alpha \vert \beta \rangle_{\mathtt{c}} = (2\pi)^{-q/2} ~ e^{i \beta^{T} \alpha } $.
\end{example}

\section{Conclusion}
This paper was designated to study quantum-mechanical observables under a symplectic transformation of coordinates. We assumed to be given a classical system with a configuration space isomorphic to $ \mathbb{R}^{q} $ (e.g. a set of finitely many coupled harmonic oscillators), and introduced in a standard manner its quantum analogue. The correspondence between the two systems was made clear. Using the rigged Hilbert space formalism, we gave correct meaning to the Dirac notation, and defined eigenstates of the quantum observables. These are coordinate-dependent, since they measure values of coordinates and momenta in a chosen symplectic basis. It is then natural to ask what happens if one chooses \textit{another symplectic basis} in the phase space, which gives rise to a new, symplectically transformed, set of observables on the Hilbert space. The main goal of this paper was to present a computation of the eigenstates of observables under such symplectic transformation. We search for them in terms of the original coordinate eigenstate basis, using the Dirac formalism.

The results are the following. Suppose that $ \hat{y}_{A} $ and $ \hat{y}_{A+q} $ with $ A = 1, ..., q $ are the coordinate and momentum observables, respectively, corresponding to the coordinates $ y_{A} $ and momenta $ y_{A+q} $ of a point $ y = e_{I} y_{I} $ in the phase space, with an implicit summation over $ I = 1, ..., 2q $. Define their eigenstates by $ \hat{y}_{A} \vert \alpha \rangle_{\mathtt{c}} = \alpha_{A} \vert \alpha \rangle_{\mathtt{c}} $ and $ \hat{y}_{A+q} \vert \beta \rangle_{\mathtt{m}} = \beta_{A} \vert \beta \rangle_{\mathtt{m}} $. Then assume the transformation
\begin{equation}\label{key}
\hat{w}_{I} = W_{IJ} ~ \hat{y}_{J}
\end{equation}
with a symplectic $ 2q \times 2q $ matrix
\begin{equation}\label{key}
W = \begin{pmatrix}
E & F \\
G & H
\end{pmatrix}
\end{equation}
and define the new eigenstates by $ \hat{w}_{A} \vert \omega \rangle_{\dot{\mathtt{c}}} = \omega_{A} \vert \omega \rangle_{\dot{\mathtt{c}}} $ and $ \hat{w}_{A+q} \vert \varrho \rangle_{\dot{\mathtt{m}}} = \varrho_{A} \vert \varrho \rangle_{\dot{\mathtt{m}}} $. Then they can be expressed in the coordinate eigenstate basis as
\begin{equation}\label{key}
\vert \omega \rangle_{\dot{\mathtt{c}}} = \int_{\mathbb{R}^{r}} d^{r}\xi ~ \vert V_{1} \xi + E^{+} U_{2} U_{2}^{T} \omega \rangle_{\mathtt{c}} ~ e^{i C} e^{i \left(-\frac{1}{2} \xi^{T} V_{1}^{T} F^{+} E F^{T} F^{+T} V_{1} \xi + \xi^{T} V_{1}^{T} F^{+}\omega \right) }
\end{equation}
\begin{equation}\label{key}
\vert \varrho \rangle_{\dot{\mathtt{m}}} = \int_{\mathbb{R}^{\acute{r}}} d^{\acute{r}}\chi ~ \vert \acute{V}_{1} \chi + G^{+} \acute{U}_{2} \acute{U}_{2}^{T} \varrho \rangle_{\mathtt{c}} ~ e^{i K} e^{i \left(- \frac{1}{2} \chi^{T} \acute{V}_{1}^{T} H^{+} G H^{T} H^{+T} \acute{V}_{1} \chi + \chi^{T} \acute{V}_{1}^{T} H^{+} \varrho \right) }
\end{equation}
Here, $ \omega \in \mathbb{R}^{q} $ is a $ q $-tuple of eigenvalues describing the eigenstate, $ r \equiv \text{rank}(F) $ and $ V_{1} \equiv V_{1}(F) $, $ U_{2} \equiv U_{2}(F) $ are matrices associated to $ F $ via the narrowed singular value decomposition. Similarly, $ \varrho \in \mathbb{R}^{q} $, $ \acute{r} \equiv \text{rank}(H) $ and $ \acute{V}_{1} \equiv V_{1}(H) $, $ \acute{U}_{2} \equiv U_{2}(H) $. We also check explicitly that upon choosing suitable normalization constants $ C, K \in \mathbb{C} $, one gets the same orthogonality relations $ _{\dot{\mathtt{c}}}\langle \varrho \vert \omega \rangle_{\dot{\mathtt{c}}} = \delta^{q}(\varrho - \omega) $ and $ _{\dot{\mathtt{m}}}\langle \varrho \vert  \omega \rangle_{\dot{\mathtt{m}}} = \delta^{q}( \varrho -  \omega ) $ as one had in the original basis. This implies the standard form of resolutions of identity.

\section*{Acknowledgments}
This work was supported by Charles University Grant Agency [Project No. 906419].

\bibliographystyle{unsrt}
\renewcommand{\bibname}{Bibliography}
\bibliography{bibliography}

\end{document}